\documentclass[conference]{IEEEtran}

\usepackage{amsfonts}
\usepackage{amssymb}
\usepackage{amsmath}
\usepackage{mathtools}
\usepackage{amsthm}
\usepackage{amstext}
\usepackage{array}
\usepackage{pifont}
\usepackage{color}
\usepackage{tikz}

\usepackage{enumitem}

\usepackage{thm-restate}
\usepackage{xspace}
\usepackage{misc}

\allowdisplaybreaks
\ifCLASSINFOpdf
\else
\fi
\begin{document}
%
\title{Living without Beth and Craig:\\ Definitions and Interpolants
in the Guarded and Two-Variable Fragments}

\author{\IEEEauthorblockN{Jean Christoph Jung}
\IEEEauthorblockA{University of Hildesheim\\
jungj@uni-hildesheim.de}
 \and
 \IEEEauthorblockN{Frank Wolter}
 \IEEEauthorblockA{University of Liverpool\\
 wolter@liverpool.ac.uk}
 }


%

\IEEEoverridecommandlockouts
\IEEEpubid{\makebox[\columnwidth]{978-1-6654-4895-6/21/\$31.00~
\copyright2021 IEEE \hfill} \hspace{\columnsep}\makebox[\columnwidth]{ }}


\newtheorem{theorem}{Theorem}
\newtheorem{lemma}{Lemma}
\newtheorem{corollary}{Corollary}
\newtheorem{definition}{Definition}
\newtheorem{example}{Example}
\newtheorem{remark}{Remark}

\IEEEoverridecommandlockouts
\IEEEpubid{\makebox[\columnwidth]{978-1-6654-4895-6/21/\$31.00~
		\copyright2021 IEEE \hfill} \hspace{\columnsep}\makebox[\columnwidth]{ }}

\maketitle

\begin{abstract} In logics with the Craig interpolation property (CIP)
  the existence of an interpolant for an implication follows from the
  validity of the implication. In logics with the projective Beth
  definability property (PBDP), the existence of an explicit
  definition of a relation follows from the validity of a formula
  expressing its implicit definability. The two-variable fragment,
  FO$^{2}$, and the guarded fragment, GF, of first-order logic both
  fail to have the CIP and the PBDP. We show that nevertheless in both
  fragments the existence of interpolants and explicit definitions is
  decidable. In GF, both problems are \ThreeExpTime-complete in
  general, and \TwoExpTime-complete if the arity of relation symbols
  is bounded by a constant $c\geq 3$.  In FO$^{2}$, we prove a
  \coNTwoExpTime upper bound and a \TwoExpTime lower bound for both
  problems. Thus, both for GF and FO$^2$ existence of interpolants and
  explicit definitions are decidable but harder than validity (in case of FO$^2$
  under standard complexity assumptions).
%


%
\end{abstract}
%

%
\IEEEpeerreviewmaketitle

\section{Introduction}
A logic enjoys the \emph{Craig Interpolation Property} (CIP) 
if an implication $\varphi \Rightarrow \psi$ is valid if, and only if, 
there exists a formula $\chi$ using only
the common symbols of $\varphi$ and $\psi$ such that
$\varphi\Rightarrow \chi$ and $\chi \Rightarrow \psi$ are both valid.
The formula $\chi$ is then called
an interpolant for $\varphi \Rightarrow \psi$. The CIP is generally regarded as one of the most
important and useful results in formal logic, with numerous
applications~\cite{van2008many,DBLP:conf/cav/McMillan03,10.1007/978-3-540-78800-3_30,TenEtAl13,DBLP:series/synthesis/2016Benedikt}. 
One particularly interesting consequence of the CIP is the \emph{Projective Beth Definability Property} (PBDP), which states that 
if a relation is implicitly definable over symbols in a 
signature $\tau$, then it is explicitly definable over $\tau$. 


From an algorithmic viewpoint, the CIP and PBDP are of interest
because they reduce existence problems to validity checking: an
interpolant \emph{exists} if, and only if, an implication is valid and an explicit definition \emph{exists} if, and only if, a straightforward formula stating implicit definability
is valid. The interpolant and explicit definition existence problems are thus not harder than validity.

In this article, we investigate the interpolant and explicit definition existence problem for two fragments of first-order logic (FO) that fail to have the CIP and PBDP: the guarded fragment (GF) and the two-variable fragment
(FO$^{2}$) of FO. GF has been introduced as a generalization of modal logic
that enjoys many of its attractive algorithmic and model-theoretic properties, including decidability, the finite model property,
the tree-like model property, and preservation properties such as the
\L{}o\'s-Tarski preservation theorem~\cite{ANvB98,DBLP:journals/jsyml/Gradel99}. Since its introduction, the guarded fragment and variants of it have been investigated extensively~\cite{DBLP:conf/lics/GradelW99,DBLP:conf/lics/GradelHO00,DBLP:journals/jacm/Otto12,DBLP:journals/jacm/BaranyCS15}, not only as a natural generalisation 
of modal logic but also in databases and knowledge representation~\cite{DBLP:journals/corr/BaranyGO13,tocl2020}. 

While GF is a good generalization of modal logic in many respects, in
contrast to modal logic it neither enjoys the
CIP~\cite{DBLP:conf/lpar/HooglandMO99} nor the
PBDP~\cite{DBLP:conf/mfcs/BaranyBC13}. Note, however, that GF  
enjoys the (non-projective) Beth Definability Property (BDP) in which
the signature $\tau$ of the implicit and explicit definitions contains
all symbols except the relation to be defined~\cite{DBLP:conf/lpar/HooglandMO99}.

Fragments of first-order logic with at most $k\geq 1$ variables have been investigated in a variety of contexts, for example in finite-model theory~\cite{Henkin,DBLP:books/cu/O2017,grohe_1998}. The two-variable fragment FO$^{2}$ is of particular interest as it is decidable (and any $k$-variable fragment with $k\geq 3$ is undecidable) and also generalizes modal logic. In fact, satisfiability of FO$^{2}$ formulas is \NExpTime-complete~\cite{DBLP:journals/bsl/GradelKV97} and FO$^{2}$ shares with modal logic and GF the finite model property. In contrast to modal logic and GF, however, it does not enjoy any tree-like model property and is less robust under extensions~\cite{DBLP:conf/dimacs/Vardi96,DBLP:books/ws/phaunRS01/Gradel01,DBLP:journals/siglog/KieronskiPT18}. 
Failure of the CIP for FO$^{2}$ was first shown using algebraic
techniques~\cite{comer1969,Pigozzi71}; more recently, an alternative
model-theoretical
proof was given~\cite{DBLP:journals/ndjfl/MarxA98}.
In contrast to GF, FO$^{2}$ does not only not enjoy the PBDP but also not the BDP~\cite{NemetiBeth2,Andreka1}.

In this article, we aim to understand better the complexity of
deciding the existence of interpolants and explicit definitions for
logics that do not enjoy the CIP and PBDP. In addition, our motivation
for investigating these existence problems in GF and FO$^{2}$ stems
from the following applications.

\medskip \noindent \emph{Strong separability of labeled data under
ontologies.} There are several scenarios in which one aims to find a
logical formula that separates positive from negative examples given
in the form of labeled data items. Examples include concept learning
in description logic~\cite{DBLP:journals/ml/LehmannH10}, reverse
engineering of database queries, also known as query by example
(QBE)~\cite{martins2019reverse}, and generating referring expressions
(GRE), where the aim is to find a formula that separates a single
positive data item from all other data
items~\cite{DBLP:journals/coling/KrahmerD12}.
In~\cite{JLPW-KR20,JLPW-DL20} an attempt is made to provide a unifying
framework for these scenarios under the assumption that the data is
given by a relational database and additional background information
is available in the form of an ontology in first-order logic. A
natural version of separability then asks whether for an ontology
$\Omc$, a database 
$\Dmc$, a signature $\tau$ of relation symbols, and sets $P$ (of
positive examples) and $N$ (of negative examples) of tuples in $\Dmc$
of the same length whether there exists a formula $\varphi$ over
$\tau$ that separates $P$ from $N$ in the sense that $\Omc\cup
\Dmc\models \varphi(\abf)$ for all $\abf\in P$, and $\Omc\cup
\Dmc\models \neg \varphi(\bbf)$, for all $\bbf\in N$. For the
fundamental cases that $\Omc$ is in GF or FO$^{2}$ and one asks for a
separating formula in GF or FO$^{2}$, respectively, it is not
difficult to see that there is a polynomial time reduction of
separability to interpolant existence. Moreover, interpolants give
rise to separating formulas and vice versa.
	
\smallskip
\noindent
\emph{Explicit definitions of relation symbols under GF and FO$^{2}$-sentences.}
The computation of explicit definitions 
of relations under ontologies has been proposed to support ontology engineering~\cite{DBLP:conf/kr/CateCMV06,TenEtAl13,DBLP:conf/ekaw/GeletaPT16}. For example, such definitions can then be included in the ontology instead of less transparent general axioms. In this application the focus shifts from interpolants to the existence of explicit definitions over a signature.

%
%
%
%
	
	\medskip

The following theorem summarizes our results:
\begin{theorem} \label{lem:lower-bounds}
	(i) The explicit GF-definability and the GF-interpolant existence problems are 
	both \ThreeExpTime-complete in general, and \TwoExpTime-complete if the arity of
	relation symbols is bounded by a constant $c\geq 3$.

	(ii) The explicit FO$^{2}$-definability and the
	FO$^{2}$-interpolant existence problems are in \coNTwoExpTime and \TwoExpTime-hard. \TwoExpTime-hardness holds already for explicit FO$^{2}$-definability using any symbol except the defined one.
\end{theorem}

For GF, it follows that interpolant and explicit definition existence
are exactly one exponential harder than validity, both in general and
if the arity of relation symbols is bounded by a constant $c\geq
3$~\cite{DBLP:journals/jsyml/Gradel99}. If the arity of symbols is
bounded by two, the corresponding fragment of GF enjoys both CIP and 
PBDP~\cite{DBLP:journals/sLogica/HooglandM02} and so interpolant and
explicit definition existence are \ExpTime-complete. 
%
Explicit
GF-definability using any symbols except the defined one is polynomial
time reducible to validity since GF has the BDP. For FO$^{2}$, it
follows that all these problems are harder than validity, unless
\coNExpTime = \TwoExpTime. Finding tight complexity bounds remains an
open problem.

The proofs start with a straightforward model-theoretic characterization
of the non-existence of an interpolant for an implication
$\varphi\Rightarrow \psi$ by the existence of appropriate
bisimulations between models satisfying $\varphi$ and $\neg \psi$,
respectively. The \emph{guarded bisimulations} used for GF were
introduced in~\cite{ANvB98} to characterize the expressive power of GF
within FO, see also~\cite{goranko20075,DBLP:books/daglib/p/Gradel014}.
The FO$^{2}$-bisimulations used for FO$^{2}$ are a
variant of the well-known pebble games characterizing finite variable logics~\cite{DBLP:journals/jsyml/Barwise77,DBLP:journals/jcss/Immerman82}. 
For GF, we then employ a mosaic-based approach, using as mosaics sets of types over $\varphi,\neg \psi$ which can be satisfied by tuples that are guarded bisimilar. Constraints for sets of such
mosaics characterize when they can be linked together to construct,
simultaneously, models of $\varphi$ and $\neg\psi$ and a guarded bisimulation between them. The triple exponential upper bound then
follows from the observation that there are triple exponentially many
mosaics. If the arity of relation symbols is bounded by a constant, then there are only double exponentially many mosaics. The lower bounds are proved by a reduction of the
word problem for languages recognized by space-bounded alternating Turing
machines.

For FO$^{2}$ we show, using mosaics that are similar to those
introduced for GF, that if there are FO$^{2}$-bisimilar models
satisfying  FO$^{2}$-formulas $\varphi$, $\neg\psi$, then there are
such models of at most double-exponential size. The \coNTwoExpTime
upper bound follows immediately from this finite model property
result. The lower bound is again proved by reduction of the word
problem for languages recognized by space-bounded alternating Turing machines.

\section{Related Work}
\label{sec:rw}

The problem of deciding the existence of explicit definitions and interpolants has hardly been studied for logics without the PBDP and CIP, respectively.
Exceptions are linear temporal logic, LTL, for which the decidability of interpolant existence has been shown in~\cite{DBLP:journals/corr/PlaceZ14,henkell1,henkell2} 
and description logics with nominals and/or role inclusions for which
2\ExpTime-completeness has recently been shown~\cite{AJMOW-AAAI21}.
Our techniques are inspired by~\cite{AJMOW-AAAI21} but are
significantly more involved.

Query determinacy and rewritability in databases can also be regarded
as explicit definability
problems~\cite{10.1145/1265530.1265534,DBLP:journals/tods/NashSV10,DBLP:conf/icdt/Marcinkowski20},
but there the focus is mainly on database query languages such as (unions) of conjunctive queries. The importance of interpolants and explicit definitions for a large variety of database applications is discussed in~\cite{DBLP:series/synthesis/2011Toman,DBLP:series/synthesis/2016Benedikt}. 

%
The guarded negation fragment of FO (GNF) extends GF by adding,
in a careful way, unions of conjunctive
queries~\cite{DBLP:journals/jacm/BaranyCS15}. It is still decidable, has
the finite model property and the tree-like model property, and
enjoys various preservation
theorems~\cite{DBLP:journals/jacm/BaranyCS15,DBLP:journals/jsyml/BaranyBC18}.
Importantly, and in contrast to GF, GNF enjoys the CIP and the
PBDP~\cite{DBLP:journals/jsyml/BaranyBC18,DBLP:journals/tocl/BenediktCB16}.
Thus, the existence of Craig interpolants and explicit definitions
reduces to validity checking which is \TwoExpTime-complete in GNF and
even in \ExpTime if the arity of relation symbols is bounded by a
constant. Thus, the existence of interpolants and explicit definitions
is one exponential harder in GF than in GNF.

Also related is work on uniform interpolation. As GF and FO$^{2}$ do not enjoy
the CIP, they also do not enjoy the uniform interpolation property
(UIP). 
In fact, uniform interpolant existence is known to be undecidable both 
for GF and FO$^{2}$~\cite{JLMSW17}, which is in contrast to the
decidability results obtained in this article for interpolant existence. 
It is also in contrast to the decidability of uniform interpolant existence problems in many standard description logics~\cite{DBLP:conf/ijcai/LutzW11,DBLP:conf/kr/LutzSW12}.
We note that GF does enjoy a `modal variant' of both the CIP and the UIP, in which besides shared symbols all
	symbols that occur in guards are allowed in the
	interpolant~\cite{HM99,DBLP:journals/sLogica/HooglandM02,DBLP:journals/tcs/DAgostinoL15}.  
	
Recently, it has been shown in \cite{andrka2020twovariable} that FO$^{2}$ 
enjoys the weak Beth definability property which requires the relation
to be explicitly defined not only to be implicitly definable but also to
exist.
Also relevant for this work is the investigation
of interpolation and definability in modal logic in general~\cite{GabMaks} and in hybrid modal logic~\cite{DBLP:journals/jsyml/ArecesBM01,DBLP:journals/jsyml/Cate05}.

\section{Preliminaries}

Let $\tau$ range over relational signatures not containing function or constant symbols. Denote by FO$(\tau)$ the set of first-order (FO) formulas constructed from
atomic formulas $x=y$ and $R(\xbf)$, $R\in \tau$, 
using conjunction, disjunction, negation, and existential and
universal quantification.
The signature $\text{sig}(\varphi)$ of an FO-formula $\varphi$ is the set of relation symbols used in it.
As usual, we write $\vp(\xbf)$ to indicate that the free variables
in $\vp$ are all from $\xbf$ and call a formula without free variables
a \emph{sentence}. 
FO$(\tau)$ is interpreted in $\tau$-structures
$ \Amf=(\text{dom}(\Amf),(R^{\Amf})_{R\in \tau})$,
where $\text{dom}(\Amf)$ is the non-empty \emph{domain} of $\Amf$,
and each $R^{\Amf}$ is a relation over $\text{dom}(\Amf)$ whose arity
matches that of $R$. We often drop $\tau$ and simply speak of structures $\Amf$. 

In the \emph{guarded fragment}, GF of
FO~\cite{ANvB98,DBLP:journals/jsyml/Gradel99}, formulas are built from
atomic formulas $R(\xbf)$ and $x=y$ by applying the Boolean
connectives and \emph{guarded quantifiers} of the form $$ \forall
\ybf(\alpha(\xbf,\ybf)\rightarrow \varphi(\xbf,\ybf)) \text{ and }
\exists \ybf(\alpha(\xbf,\ybf)\wedge \varphi(\xbf,\ybf)) $$ where
$\varphi(\xbf,\ybf)$ is a guarded formula, and $\alpha(\xbf,\ybf)$ is
an atomic formula that contains all variables in $\xbf,\ybf$. The
formula $\alpha$
is called the \emph{guard of the quantifier}.  GF$(\tau)$ denotes the
set of all guarded formulas (also called GF-formulas) over
signature~$\tau$.  We regard $\forall\ybf(\alpha(\xbf,\ybf)\rightarrow
\varphi(\xbf,\ybf))$ as an abbreviation for $\neg\exists
\ybf(\alpha(\xbf,\ybf)\wedge\neg \varphi(\xbf,\ybf))$.  The
\emph{two-variable fragment}, FO$^{2}$, of FO consists of all formulas 
in FO 
using two distinct variables. 

Let $\Amf$ be structure. A pair $\Amf,\abf$ with $\abf$ a tuple in $\Amf$ is called a \emph{pointed structure}. It will be convenient to use the notation $[\abf]=\{a_{1},\ldots,a_{n}\}$ to denote the set of components of the tuple $\abf=(a_{1},\ldots,a_{n})\in \text{dom}(\Amf)^{n}$.
Similarly, for a tuple $\xbf = (x_{1},\ldots,x_{n})$ of variables we use $[\xbf]$ to denote the set $\{x_{1},\ldots,x_{n}\}$.

We next recall model-theoretic characterizations of when pointed structures cannot be distinguished in either GF or FO$^{2}$. 
We begin by introducing GF($\tau$)-bisimulations (often called guarded $\tau$-bisimulations)~\cite{DBLP:books/daglib/p/Gradel014}.
A set $G \subseteq \text{dom}(\Amf)$ is \emph{$\tau$-guarded} in $\mathfrak{A}$ if $G$ is a singleton or 
there exists $R\in \tau$ with $\Amf\models R(\abf)$ such that $G = [\abf]$. 
A tuple $\abf \in \text{dom}(\Amf)^{n}$ is
\emph{$\tau$-guarded} in $\mathfrak{A}$ if $[\abf]$ is a subset of some $\tau$-guarded set in~$\mathfrak{A}$. 

For tuples $\abf=(a_{1},\ldots,a_{n})$ in $\mathfrak{A}$
and $\bbf=(b_{1},\ldots,b_{n})$ in $\mathfrak{B}$ we call a mapping $p$ from $[\abf]$ to
$[\bbf]$ with $p(a_{i})=b_{i}$ for $1\leq i \leq n$ (written $p:\abf\mapsto \bbf$)
a \emph{partial $\tau$-isomorphism} if $p$ is an isomorphism from the $\tau$-reduct of $\mathfrak{A}_{|[\abf]}$ onto
$\mathfrak{B}_{|[\bbf]}$, where $\Amf_{|X}$ denotes the restriction of a structure $\Amf$ to a subset $X$ of its domain.

A set $I$ of partial $\tau$-isomorphisms $p: \abf \mapsto \bbf$ 
from $\tau$-guarded tuples $\abf$ in $\mathfrak{A}$ to $\tau$-guarded tuples $\bbf$ in $\Bmf$ is a 
\emph{GF($\tau$)-bisimulation} if the following hold for all 
$p: \abf \mapsto \bbf\in I$:
\begin{enumerate}
	
	\item[(i)] for every $\tau$-guarded tuple $\abf'$ in $\Amf$ there exists a $\tau$-guarded
	tuple $\bbf'$ in $\Bmf$ and $p': \abf'\mapsto \bbf'\in I$
	such that $p'$ and $p$ coincide on $[\abf]\cap [\abf']$.
	
	\item[(ii)] for every $\tau$-guarded tuple $\bbf'$ in $\Bmf$ there exists a $\tau$-guarded
	tuple $\abf'$ in $\Amf$ and $p': \abf'\mapsto \bbf'\in I$
	such that $p'^{-1}$ and $p^{-1}$ coincide on $[\bbf]\cap
	[\bbf']$.
	
\end{enumerate}
Assume that $\abf$ and $\bbf$ are (possibly not $\tau$-guarded) tuples in $\Amf$ and $\Bmf$. Then we say that the pointed structures $\Amf,\abf$ and $\Bmf,\bbf$ are \emph{GF($\tau$)-bisimilar},
in symbols $\Amf,\abf \sim_{\text{GF},\tau} \Bmf,\bbf$,
if there exists a partial $\tau$-isomorphism 
$p: \abf\mapsto \bbf$ and a GF($\tau$)-bisimulation $I$ such that Conditions~(i) and~(ii) hold for $p$. 

Next we introduce appropriate bisimulations for FO$^2$, which are
essentially a relational variant of 
the infinite $2$-pebble games which have been used to
characterize the expressive power of FO$^2$, see e.g.~\cite{DBLP:journals/jsyml/Barwise77,DBLP:journals/jcss/Immerman82}.
Given structures $\Amf,\Bmf$, a relation $S\subseteq
\text{dom}(\Amf)\times\text{dom}(\Bmf)$ is an
\emph{FO$^2(\tau)$-bisimulation between \Amf and \Bmf} if $S$ is
\emph{global}, that is, $\text{dom}(\Amf)\subseteq \{a\mid (a,b)\in
S\}$ and $\text{dom}(\Bmf)\subseteq \{b\mid (a,b)\in S\}$ and, for
every $(a,b)\in S$ the following conditions are satisfied:
\begin{itemize}
	

    \item[(i)] for every $a'\in\text{dom}(\Amf)$, there is a
    $b'\in \text{dom}(\Bmf)$ such that $(a,a')\mapsto (b,b')$ is a
    partial $\tau$-isomorphism between \Amf and \Bmf and $(a',b')\in
    S$;
    
    \item[(ii)] for every $b'\in\text{dom}(\Bmf)$, 
    there is a $a'\in \text{dom}(\Amf)$ such that 
    $(a,a')\mapsto (b,b')$ is a
    partial $\tau$-isomorphism between \Amf and \Bmf and $(a',b')\in
    S$.
\end{itemize}
For tuples $\abf=(a_1,\ldots,a_n),\bbf=(b_1,\ldots,b_n)$ of equal length
$n=0,1,2$, we write $\Amf,\abf\sim_{\text{FO}^2,\tau}\Bmf,\bbf$ iff
$\abf\mapsto\bbf$ is a partial $\tau$-isomorphism between \Amf and
\Bmf and there is an FO$^2(\tau)$-bisimulation $S$ between $\Amf$ and
\Bmf such that $(a_i,b_i)\in S$, for all $i\leq n$.

Now, let $\Lmc$ be either GF or FO$^{2}$. We write $\Amf,\abf
\equiv_{\Lmc,\tau} \Bmf,\bbf$
and call $\Amf,\abf$ and $\Bmf,\bbf$ \emph{$\Lmc(\tau)$-equivalent} if
$\Amf\models \varphi(\abf)$ iff $\Bmf\models \varphi(\bbf)$ holds for
all formulas $\varphi$ in $\Lmc(\tau)$. The following equivalences are
well-known~\cite{goranko20075,DBLP:books/daglib/p/Gradel014}.
\begin{lemma}\label{lem:guardedbisim}
	Let $\Lmc$ be either GF or FO$^{2}$. Let $\Amf,\abf$ and $\Bmf,\bbf$ be pointed structures and $\tau$ a signature. Then
	$$
	\Amf,\abf \sim_{\Lmc,\tau} \Bmf,\bbf \quad \text{ implies } \quad
	\Amf,\abf \equiv_{\Lmc,\tau} \Bmf,\bbf
	$$
	and, conversely, if $\Amf$ and $\Bmf$ are $\omega$-saturated, then
	$$
	\Amf,\abf \equiv_{\Lmc,\tau} \Bmf,\bbf \quad \text{ implies } \quad 
	\Amf,\abf \sim_{\Lmc,\tau} \Bmf,\bbf
	$$
\end{lemma}

%
%
%
%
%
\section{Interpolants and Explicit Definitions}
\label{sec:one}
Let $\Lmc$ be either GF or FO$^{2}$. We introduce $\Lmc$-interpolants
and explicit $\Lmc$-definitions and provide model-theoretic
characterizations of the existence of $\Lmc$-interpolants and explicit
$\Lmc$-definitions using $\Lmc$-bisimulations.

Let $\varphi(\xbf),\psi(\xbf)$ be $\Lmc$-formulas
with the same free variables $\xbf$.
We call an $\Lmc$-formula $\theta(\xbf)$ an \emph{$\Lmc$-interpolant
for $\varphi,\psi$} if $\text{sig}(\theta) \subseteq \text{sig}(\varphi)\cap \text{sig}(\psi)$, 
$\varphi(\xbf)\models \theta(\xbf)$, and  $\theta(\xbf)\models \psi(\xbf)$. 
We are interested in \emph{$\Lmc$-interpolant existence},
the problem to decide for given
$\varphi(\xbf),\psi(\xbf)$ in $\Lmc$ whether an $\Lmc$-interpolant for 
$\varphi(\xbf), \psi(\xbf)$ exists. Recall from the introduction that neither GF nor FO$^{2}$ enjoy the \emph{Craig Interpolation Property} (CIP) according to which an $\Lmc$-interpolant for $\Lmc$-formulas $\varphi(\xbf),\psi(\xbf)$ exists iff $\varphi(\xbf)\models \psi(\xbf)$.

Following Robinson's approach to interpolation and definability~\cite{AbRobinson}, we call $\Lmc$-formulas $\varphi(\xbf),\psi(\xbf)$ \emph{jointly
$\Lmc(\tau)$-consistent} if there exist pointed structures $\Amf,\abf$ and
$\Bmf,\bbf$ with $\Amf\models \varphi(\abf)$ and $\Bmf \models
\psi(\bbf)$ such that $\Amf,\abf \sim_{\Lmc,\tau} \Bmf,\bbf$. The
notion of joint consistency has been used to explore interpolation
properties in finite variable infinitary
logics~\cite{DBLP:journals/jsyml/BarwiseB99}
and (implicitly) to show the
lack of the CIP for GF~\cite{DBLP:journals/sLogica/HooglandM02}. Using Lemma~\ref{lem:guardedbisim} we show that
 interpolant existence can in fact be characterized via joint
consistency.
%
%
%
\begin{restatable}{lemma}{lemcharinterpolation}
\label{lem:gf-char-interpolation}
  Let $\Lmc$ be either FO$^{2}$ or GF. Let $\varphi(\xbf),\psi(\xbf)$
  be $\Lmc$-formulas and let
  $\tau= \text{sig}(\varphi) \cap \text{sig}(\psi)$. Then the following conditions are equivalent:
\begin{enumerate}
  	\item there does not
  exist an $\Lmc$-interpolant for $\varphi(\xbf),\psi(\xbf)$;
    \item $\varphi(\xbf),\neg \psi(\xbf)$ are jointly $\Lmc(\tau)$-consistent.
\end{enumerate}  %
\end{restatable}
The following example illustrates the introduced notions.
\begin{example} \label{ex:joint}\emph{Consider the GF-formulas
  $\varphi(x),\psi(x)$ given by
  \begin{align*}
    \varphi(x) & = \exists yz\, (G(x,y,z)\wedge R(x,y)\wedge
    R(y,z)\wedge R(z,x))
    \\
    \psi(x) & = A(x) \wedge \forall y\forall z\,(R(y,z)\rightarrow
    (A(y)\leftrightarrow \neg A(z)))
  \end{align*}
  Clearly, we have $\varphi(x)\models\neg \psi(x)$.
  Moreover, the models $\Amf,a$ of $\varphi(x)$ and $\Bmf,b$ of $\psi(x)$
  depicted in Fig.~\ref{fig:models} witness that $\varphi(x)$ and
  $\psi(x)$ are jointly $\text{GF}(\{R\})$-consistent; in fact, the
  witnessing GF-bisimulation contains $n\mapsto m$, for every
  $n\in \text{dom}(\Amf)$, $m\in \text{dom}(\Bmf)$,
  and
  $\abf\mapsto\bbf$, for all $\abf\in R^\Amf$, $\bbf\in R^\Bmf$.
  Lemma~\ref{lem:gf-char-interpolation} implies that there
  is no $\text{GF}$-interpolant for
$\varphi(x),\neg\psi(x)$.}\hfill$\dashv$
\end{example}

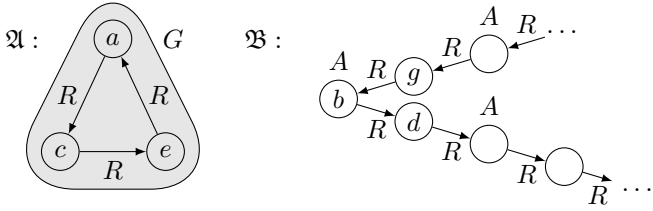
\begin{figure}[t]
  \centering
  \begin{tikzpicture}
    \draw [rounded corners=10mm,fill=gray!20] (0,0)--(3,0)--(1.5,3)--cycle;

    \node (labela) at (0.3,2) {$\Amf:$};
    \node at (2.3,2) {$G$};
    \node[circle,draw,minimum size=5mm,inner sep=0mm] (a) at (1.5,2.0) {$a$};
    \node[circle,draw,minimum size=5mm,inner sep=0mm] (b) at (.8,.5)
    {$c$};
    \node[circle,draw,minimum size=5mm,inner sep=0mm] (c) at (2.2,.5)
    {$e$};
    
    \draw[-latex] (a) to node[left] {$R$} (b);
    \draw[-latex] (b) to node[below] {$R$} (c);
    \draw[-latex] (c) to node[right] {$R$} (a);

    \node (labela) at (3.5,2) {$\Bmf:$};

    \node[circle,draw,label=above:{$A$},minimum size=5mm,inner
    sep=0mm] (d) at (4.5,1.2) {$b$};
    \node[circle,draw,minimum size=5mm,inner sep=0mm] (e) at (5.5,.9)
    {$d$};
    \node[circle,draw,label=above:{$A$},minimum size=5mm,inner
    sep=0mm] (f) at (6.5,.6) {\phantom{$g$}};
    \node[circle,draw,minimum size=5mm,inner sep=0mm] (g) at (7.5,.3)
    {$\phantom{h}$};
    \node (h) at (8.5,0) {$\dots$};

    \node[circle,draw,minimum size=5mm,inner sep=0mm] (e1) at (5.5,1.5)
    {$g$};
    \node[circle,draw,label=above:{$A$},minimum size=5mm,inner
    sep=0mm] (f1) at (6.5,1.8) {$\phantom{g}$};
    \node[minimum size=5mm,inner sep=0mm] (g1) at (7.5,2.1)
    {$\dots$};

    \draw[-latex] (d) to node[below] {$R$} (e);
    \draw[-latex] (e) to node[below] {$R$} (f);
    \draw[-latex] (f) to node[below] {$R$} (g);
    \draw[-latex] (g) to node[below] {$R$} (h);

    \draw[-latex] (e1) to node[above] {$R$} (d);
    \draw[-latex] (f1) to node[above] {$R$} (e1);
    \draw[-latex] (g1) to node[above] {$R$} (f1);
 
  \end{tikzpicture}
%
%
%
%
%
%
%
  \caption{Models for formulas in Example~\ref{ex:joint}}
  \label{fig:models}
\end{figure}

Let $\vp$ be an $\Lmc$-sentence, $\theta(\xbf)$ an $\Lmc$-formula, and
$\tau$ a signature. An $\Lmc(\tau)$-formula $\psi(\xbf)$ is an
\emph{explicit $\Lmc(\tau)$-definition of $\theta$ under $\vp$} if
$\varphi\models \forall \xbf(\theta(\xbf)\leftrightarrow \psi(\xbf))$.
We call $\theta$ \emph{explicitly $\Lmc(\tau)$-definable under
$\varphi$} if such an explicit $\Lmc(\tau)$-definition of $\theta$
under $\vp$ exists.  We call $\theta$ \emph{implicitly
$\Lmc(\tau)$-definable under $\varphi$} if
$\varphi\wedge\varphi'\models \forall \xbf
(\theta(\xbf)\leftrightarrow \theta'(\xbf))$, where $\varphi'$ and
$\theta'$ are obtained from $\varphi$ and $\theta$, respectively, by
renaming all non-$\tau$ symbols $R$ to fresh $R'$ of the same arity.
Obviously, explicit $\Lmc(\tau)$-definability implies implicit
$\Lmc(\tau)$-definability. Recall from the introduction that neither GF nor FO$^{2}$ enjoy the \emph{projective Beth definability property} (PBDP) according to which the converse implication holds.

We consider the problem of \emph{explicit
$\Lmc$-definability}, that is, the problem to decide for given
$\vp,\theta(\xbf),\tau$ whether there is an explicit
$\Lmc(\tau)$-definition of $\theta(\xbf)$ under $\vp$. We first
observe that explicit definition existence reduces to interpolant
existence.

\begin{restatable}{lemma}{lemdeftoint}
\label{lem:def-to-int}
  Let $\Lmc$ be either FO$^{2}$ or GF. There is a polynomial time
  reduction of explicit $\Lmc$-definability to
  $\Lmc$-interpolant existence. 
\end{restatable}

%

Lemma~\ref{lem:def-to-int} suggests that there is a characterization
of explicit definability in terms of joint
$\Lmc(\tau)$-consistency as well. Indeed, we give this
characterization next. 

\begin{lemma} \label{lem:gf-char-explicit}
  Let $\Lmc$ be either FO$^{2}$ or GF. For every $\Lmc$-sentence
  $\varphi$, every $\Lmc$-formula $\theta(\xbf)$,
  and signature $\tau$, the following conditions are equivalent: 
  \begin{enumerate}

    \item there does not exist an explicit
      $\Lmc(\tau)$-definition of $\theta(\xbf)$ under
      $\varphi$; 

    \item $\varphi\wedge \theta(\xbf)$ and
      $\varphi\wedge\neg\theta(\xbf)$ are jointly
      $\Lmc(\tau)$-consistent. 
      
  \end{enumerate}
\end{lemma}
Let us also illustrate the failure of the projective Beth definability
property in FO$^2$ using Lemma~\ref{lem:gf-char-explicit}.
\begin{example}\label{ex:fo2-beth}
  \emph{Consider the FO$^2$-sentence $\varphi$ given by
  \begin{align*}
    \varphi ={} & \forall xy\, ((Y(x)\wedge Y(y)) \rightarrow
    x=y) \wedge{} \\
    & \forall x\, (Z(x)\rightarrow \textstyle \bigvee_{i=0}^3
    (\varphi_i(x) \wedge \varphi'_{3-i}(x) ) )\wedge{} \\
    & \forall xy\, ((R(x,y) \wedge \neg Z(x))\rightarrow I(x))\wedge{}
    \\
    & \forall xy\, (R(x,y) \rightarrow (I(x)\leftrightarrow
    I(y)))\wedge{} \\
    & \forall xy\, (R(x,y)\rightarrow
    (I(x)\rightarrow (A(x)\leftrightarrow \neg A(y))))
  \end{align*}
  where $\varphi_i(x)$ (resp., $\varphi_i'(x)$) is an FO$^2$-formula
  expressing that there is an $R$-path of length $i$ to (resp., from)
  an element satisfying $Y$. Observe that $Z(x)$ is implicitly
  FO$^2(\{R\})$-definable under $\varphi$ since it is explicitly
  FO$(\{R\})$-definable under $\varphi$: $Z(x)$ is true at $\Amf,a$
  iff $a$ lies on a cycle of length three. In particular, the last
  three conjuncts of $\varphi$ imply that $\neg Z(x)$ cannot be
  satisfied on any node of a cycle of odd length. To see the lack of an
  explicit FO$^2(\{R\})$-definition, consider structures $\Amf'$
  and $\Bmf'$ obtained from \Amf, \Bmf in Figure~\ref{fig:models}:}
  \begin{itemize}

    \item \emph{$\Bmf'$ is the extension of \Bmf in which every node
      satisfies $I$;}

    \item \emph{$\Amf'$ is the disjoint union of $\Bmf'$ and the extension
      of \Amf in which every node satisfies $Z$ and $a$ satisfies
    $Y$.}

  \end{itemize}
  \emph{It can be verified that $\Amf',a$ is a model of $\varphi\wedge Z(x)$, that $\Bmf',b$ is a
  model of $\varphi\wedge\neg Z(x)$, and that
  $\Amf',a\sim_{\text{FO}^2,\{R\}}\Bmf',b$. By
  Lemma~\ref{lem:gf-char-explicit}, $Z(x)$ is not explicitly
  FO$^2(\{R\})$-definable under~$\varphi$.\hfill$\dashv$}
\end{example}
%
%
%
%
%

%
%
%
%

\section{Deciding Joint GF$(\tau)$-Consistency}\label{sec:gf}
We prove Theorem~\ref{lem:lower-bounds} (i).
As Lemma~\ref{lem:gf-char-interpolation} provides a reduction of
the complement of GF-interpolant existence to joint GF$(\tau)$-consistency, that is, the
problem of deciding whether given $\varphi(\xbf),\psi(\xbf)$ are
jointly GF$(\tau)$-consistent, we will prove the complexity upper bound for the latter problem. For the complexity lower bounds, we will also
consider joint GF$(\tau)$-consistency, but for an input of the form
given in Lemma~\ref{lem:gf-char-explicit}. This yields the respective
lower bounds for explicit definition existence; by
Lemma~\ref{lem:def-to-int}, they lift to interpolant existence.
\subsection{Upper Bounds}
To decide joint GF$(\tau)$-consistency we pursue a mosaic approach
based on types. Throughout the section, let
$\varphi(\xbf_0),\psi(\xbf_0)$ be the input to joint
GF$(\tau)$-consistency, for some signature $\tau$. Let $\Xi=\{\varphi(\xbf_{0}),\psi(\xbf_{0})\}$. 

We begin by defining an appropriate notion of type.
Let $\text{width}(\Xi)$ denote the maximal arity of any relation symbol
used in $\Xi$ and let $\text{fv}(\Xi)$ be the number of variables
in $\xbf_0$. Let $x_1,\ldots,x_{2n}$ be fresh variables, where $n:=\max{\{\text{width}(\Xi),\text{fv}(\Xi)\}}$. We use $\cl(\Xi)$ to
denote the smallest set of GF-formulas that is closed under taking
subformulas and single negation, and contains: 
\begin{itemize}
	\item $\Xi$, 
    \item all formulae $x=y$ for distinct variables $x,y$;
    \item all formulae $\exists \xbf R(\xbf\ybf)$, where $R$ is a relation symbol that occurs in $\Xi$ and $\xbf\ybf$ is a tuple of  variables.
\end{itemize}
Let $\Amf$ be a structure, $\abf$ a tuple of distinct elements
from the domain of \Amf, and $\xbf$ a tuple of distinct variables in $\{x_{1},\ldots,x_{2n}\}$ of the same length as $\abf$. Consider the bijection $v: \xbf \mapsto \abf$. Then the \emph{$\Xi$-type of $\abf$ in \Amf defined
through $v$} is  
\[
\text{tp}(\Amf,v: \xbf \mapsto \abf)= \{ \theta\mid \Amf\models_{v} \theta, \theta\in \cl(\Xi)[\xbf]\},
\]
where $\cl(\Xi)[\xbf]$ is obtained from $\cl(\Xi)$ by
substituting in any formula $\theta\in \cl(\Xi)$ the free
variables of $\theta$ by variables in $[\xbf]$ in all possible ways. 
Note that the assumption that $v$ is bijective entails
that $\neg (x=y)\in \text{tp}(\Amf,v: \xbf \mapsto \abf)$ for any
two distinct $x,y\in [\xbf]$. 
We drop $v$ (and both $v$ and $\xbf$) and write 
$\text{tp}(\Amf,\xbf \mapsto \abf)$
(and $\text{tp}(\Amf,\abf)$, respectively), whenever they are obvious 
from the context. Any $\Xi$-type of some $\abf$ through some $v:\xbf\mapsto \abf$ is 
called a \emph{$\Xi$-type} and simply denoted $t(\xbf)$. 
The set of all $\Xi$-types is denoted $T(\Xi)$.

We give a high-level description of our approach. To decide joint GF$(\tau)$-consistency of
$\varphi(\xbf_0),\psi(\xbf_0)$ we determine all sets $\Phi\subseteq
T(\Xi)$ using at most $n$ variables from $\{x_{1},\ldots,x_{2n}\}$
that can be satisfied in GF($\tau$)-bisimilar models in the
following sense: there are models $\Amf_{t}$, $t\in \Phi$, realizing
$t$ in tuples $\abf_{t}$ in $\text{dom}(\Amf_{t})$ through assignments $v_{t}$ such that for any $t_{1},t_{2}\in \Phi$,
$$
\Amf_{t_{1}},v_{t_{1}}(\xbf_{t_{1},t_{2}}) \sim_{\text{GF},\tau}
\Amf_{t_{2}},v_{t_{2}}(\xbf_{t_{1},t_{2}}),
$$
where $\xbf_{t_{1},t_{2}}$ are the shared free variables of $t_{1}$ and $t_{2}$.
Such sets $\Phi$ will be called \emph{$\tau$-mosaics}. Given the set of all $\tau$-mosaics 
one can check whether $\varphi(\xbf_{0}),\psi(\xbf_{0})$ are jointly GF$(\tau)$-consistent
by simply checking whether there are types $t_{1}(\xbf),t_{2}(\xbf)$
in a single $\tau$-mosaic $\Phi$ such that one can replace the variables $\xbf_{0}$ in $\varphi(\xbf_{0}),\psi(\xbf_{0})$ by variables in $[\xbf]$ in such a way that $\varphi'\in t_{1}(\xbf_{1})$, $\psi'\in
t_{2}(\xbf_{2})$ for the resulting formulas $\varphi',\psi'$.
Thus, in what follows we aim to determine the characteristic properties 
of $\tau$-mosaics and show that they can be enumerated in triple 
exponential time in general. If $\text{width}(\Xi)$ is fixed, we
perform a closer analysis of the set of mosaics and show that double
exponential time is sufficient. The characteristic properties
of $\tau$-mosaics consist of \emph{internal} properties that can be checked by inspecting a single set $\Phi$ of $\Xi$-types in isolation and one \emph{external} property stating the existence of other $\tau$-mosaics that ensure that $\tau$-mosaics can be attached to each other in such a way that GF($\tau$)-bisimilar models can be constructed.  
%
 
To formulate the properties of $\tau$-mosaics, we
require some notation.  The \emph{restriction $t(\xbf)_{|X}$} of a
$\Xi$-type $t(\xbf)$ to a set $X$ of variables is the set of
$\theta\in t(\xbf)$ with free variables among $X$. The
\emph{restriction $\Phi_{|X}$} of a set $\Phi$ of $\Xi$-types to $X$
is defined as $\{ t(\xbf)_{|X} \mid t(\xbf)\in \Phi\}$.  Types
$t(\xbf)$ and $t'(\xbf')$ \emph{coincide on $X$} if $t(\xbf)_{|X} =
t'(\xbf')_{|X}$ and sets $\Phi,\Phi'$ of $\Xi$-types \emph{coincide on
$X$} if $\Phi_{|X}= \Phi'_{|X}$. A variable $x$ is \emph{free} in a
mosaic $\Phi$ if $\Phi$ contains a type in which $x$ is free.

A formula $Q(\xbf)$ of the form $x=x$ or $\exists \ybf R(\xbf\ybf)$ with $R\in \tau$ is called a \emph{$\tau$-guard (for $\xbf$)}. It is called a \emph{strict $\tau$-guard} if it is of the form $x=x$ or $\ybf$ is empty, respectively. We call a set $\Phi\subseteq T(\Xi)$ a \emph{$\tau$-mosaic} if
it satisfies the following conditions:
\begin{itemize}

  \item $\Phi$ is \emph{$\tau$-uniform}: for all $\tau$-guards
    $Q(\zbf)$ and all $t(\xbf),s(\ybf)\in \Phi$ with $[\zbf]\subseteq
    [\xbf]\cap [\ybf]$, $Q(\zbf)\in t(\xbf)$ iff $Q(\zbf)\in
    s(\ybf)$;

    \item \emph{closed under restrictions}: if $t(\xbf)\in \Phi$
      and $X\subseteq [\xbf]$, then $t(\xbf)_{|X} \in \Phi$; 

    \item \emph{GF($\tau$)-bisimulation saturated}: for all
      $t(\xbf)\in \Phi$, all strict $\tau$-guards $Q(\ybf)\in
      t(\xbf)$, and all $t'(\zbf)\in \Phi$ with $[\zbf]\subseteq
      [\ybf]$, there is an $s(\ybf')\in \Phi$ such that
      $t'(\zbf)\subseteq s(\ybf')$ and $[\ybf']=[\ybf]$.
\end{itemize}        
Intuitively, $\tau$-uniformity reflects that GF($\tau$)-bisimulations
preserve all $\tau$-guards and GF($\tau$)-bisimulation saturatedness
reflects Condition~(i) for GF($\tau$)-bisimulations. 
Let us illustrate how to read off a mosaic from jointly consistent
structures.
\begin{example}\label{ex:internal}
  \emph{
  Let $\Amf,\Bmf$ be the structures from Fig.~\ref{fig:models}, set
  $\tau = \{R\}$, and $\Xi = \{\varphi(x),\psi(x)\}$ 
  with $\varphi,\psi$ as in Example~\ref{ex:joint}.
  Let $\Phi$ be the closure under restrictions of the set containing
  \[\text{tp}(\Amf,xyz\mapsto ace)\text{ and all types
  }\text{tp}(\Bmf,\xbf\mapsto \bbf)\]
  with $\xbf\in\{xy,yz,zx\}$ and $\bbf\in \{gb,bd\}$. 
  Thus, for
  example, $\Phi$ contains $\text{tp}(\Amf,xy\mapsto ac)$ as well.
  It can easily be verified that $\Phi$ is $\tau$-uniform. To illustrate $\text{GF}(\tau)$-bisimulation
  saturation, consider the types
  $t(x,y)=\text{tp}(\Bmf,xy\mapsto bd)$ and $t'(x)=\text{tp}(\Amf,x\mapsto
  a)$, and the strict
  $\tau$-guard $R(x,y)$ contained in $t(x,y)$. Then
  GF$(\tau)$-bisimulation saturatedness is witnessed by the type
  $s(x,y)=\text{tp}(\Amf,xy\mapsto ac)\in \Phi$. \hfill$\dashv$}
\end{example}
In addition to the internal properties above, we have to ensure that $\tau$-mosaics can be linked together. The next two conditions state when this is the
case. 
We say that $\tau$-mosaics $\Phi_{1},\Phi_{2}$ are \emph{compatible} 
if for $\{i,j\}=\{1,2\}$:
\begin{enumerate} 

  \item for every $t(\xbf)\in\Phi_i$ there is an
    $s(\ybf)\in \Phi_j$ such that $t(\xbf)$ and $s(\ybf)$ coincide
    on $[\xbf]\cap[\ybf]$; 

  \item if there are $t(\xbf)\in \Phi_{i}$ and
    $s(\ybf)\in \Phi_{j}$ and a $\tau$-guard $Q(\zbf)\in t(\xbf)$ with
    $[\zbf]\subseteq [\xbf]\cap [\ybf]$, then $\Phi_{i}$ and
    $\Phi_{j}$ coincide on $[\zbf]$.

\end{enumerate}

Note that compatibility is a reflexive and symmetric relation.
Let $\mathcal{M}$ be a set of $\tau$-mosaics. We call 
$\Phi\in \mathcal{M}$ \emph{existentially saturated in $\mathcal{M}$} 
if for every $t(\xbf)\in \Phi$ and every formula $\exists \ybf (R(\xbf',\ybf)\wedge
\lambda(\xbf',\ybf))\in t(\xbf)$ there is a $\Phi'\in \Mmc$
such that $\Phi,\Phi'$ are compatible and 
$R(\xbf',\ybf')\wedge\lambda(\xbf',\ybf')\in t'(\zbf)$ for some
$t'(\zbf)\in \Phi'$ which coincides with $t(\xbf)$ on $[\xbf] \cap
[\zbf]$.
$\mathcal{M}$ is called \emph{existentially saturated} if every $\Phi\in \mathcal{M}$ is existentially saturated in $\mathcal{M}$.

\begin{example}\label{ex:external}
  \emph{Let $\Mmc=\{\Phi\}$ with $\Phi$ as in Example~\ref{ex:internal}.  We
  claim that $\Mmc$ is existentially saturated. Clearly every existentially quantified formula in (any restriction of) $\text{tp}(\Amf,xyz\mapsto ace)$ is
  "realized" in $\text{tp}(\Amf,xyz\mapsto ace)$ itself. Consider now, for example,
  $\exists z' R(z,z') \in t(y,z):=\text{tp}(\Bmf,yz\mapsto bd)$. Then the type $\text{tp}(\Bmf,zx \mapsto gb)$ coincides with $t(y,z)$ on $\{z\}$ and contains $R(z,x)$, as required.\hfill$\dashv$}
  %
\end{example}

It should be clear that the set of existentially saturated sets of
$\tau$-mosaics is closed under unions. Thus, the union of all
existentially  saturated sets of $\tau$-mosaics is again existentially
saturated. This set can be obtained by a purely syntactic elimination
procedure, starting with the set of all $\tau$-mosaics with at most
$n$ free variables from $\{x_{1},\ldots,x_{2n}\}$. We fine-tune and
analyze this procedure below to obtain our two complexity upper bounds for joint GF($\tau$)-consistency. To this end, we prove three lemmas about existentially saturated sets of $\tau$-mosaics. The first lemma states that $\tau$-mosaics that are
contained in an existentially saturated set behave in the way announced 
in the high-level overview of the proof. 
 
\begin{lemma}\label{lem:construct}
Assume $\mathcal{M}$ is an existentially saturated set of $\tau$-mosaics and let $t_{1}(\xbf_{1}),t_{2}(\xbf_{2})\in \Psi\in \mathcal{M}$. Then there are pointed models $\Amf_{1},\abf_{1}$
and $\Amf_{2},\abf_{2}$ and $v_{i}:\xbf_{i}\mapsto \abf_{i}$ 
such that
 \begin{itemize}
 \item 	$\Amf_{i}\models t_{i}(\abf_{i})$, $i=1,2$, and
 \item $\Amf_{1},v_{1}([\xbf_{1}]\cap [\xbf_{2}]) \sim_{\text{GF},\tau} \Amf_{2},v_{2}([\xbf_{1}]\cap [\xbf_{2}])$.
 \end{itemize}
\end{lemma}

\begin{proof}
  Let $\Psi\in \mathcal{M}$. We assume w.l.o.g. that $\mathcal{M}$ 
  is \emph{closed under restrictions} in the sense that for any $\Phi\in \mathcal{M}$ and subset $X$
  of the free variables of $\Phi$, $\Phi_{|X}\in \mathcal{M}$.
  (If it is not closed under restrictions simply add all $\Phi_{|X}$ with $\Phi\in \mathcal{M}$ to $\mathcal{M}$. The resulting set is still existentially saturated.)
  Define
  $\widehat{\Psi}:=\Psi_{|\emptyset}$, that is, $\widehat{\Psi}$
  contains all $\Xi$-types in $\Psi$ without free variables.  By
  closure under restrictions of $\mathcal{M}$, we have
  $\widehat{\Psi}\in \mathcal{M}$. Assume
  $\widehat{\Psi}=\{\hat{t}_{1},\ldots,\hat{t}_{m}\}$. 
We construct structures $\Amf_{i}$, $i=1,\ldots,m$, with $\Amf_{i}$
satisfying $\hat{t}_{i}$.  For the construction, it is useful to
employ notation for tree decompositions.  A \emph{tree decomposition}
of a structure $\mathfrak{A}$ is a triple $(T,E,\text{bag})$ with
$(T,E)$ a tree and $\text{bag}$ a function that assigns to every $t\in
T$ a set $\text{bag}(t)\subseteq \text{dom}(\mathfrak{A})$ such that
\begin{enumerate}

  \item $\mathfrak{A} = \bigcup_{t\in T}\mathfrak{A}_{|\text{bag}(t)}$;

  \item $\{t \in T\mid a\in \text{bag}(t)\}$ is connected in $(T,E)$,
    for every $a\in \text{dom}(\mathfrak{A})$.

\end{enumerate} 
%
We construct the structures $\Amf_i$, $i=1,\ldots,m$ by giving 
a tree decomposition $(T_i,E_i,\text{bag}_i)$ of $\Amf_{i}$. To this end,
we define $(T_i,E_i,\text{bag}_i)$ and structures $\text{Bag}_{i}(t)$ with domain $\text{bag}_{i}(t)$, $t\in T_{i}$, and then show that $(T_i,E_i,\text{bag}_i)$ is a tree decomposition of the union $\Amf_{i}$ of all $\text{Bag}_{i}(t)$, $t\in T_{i}$.
We start with the definition of $(T_i,E_i)$. Let $T_i$ be the set of all sequences
\[
\sigma_{n}= (t_{0}(\ybf_0),\Phi_{0}),\ldots,(t_{n}(\ybf_{n}),\Phi_{n})
\]
such that $t_0=\hat t_i$ (thus $\ybf_{0}$ is empty), $\Phi_0=\widehat{\Psi}$,
$t_j(\ybf_j)\in\Phi_j\in \mathcal{M}$ for all $j\leq n$, and for all $j<n$: 
\begin{itemize}

  \item $\Phi_{j},\Phi_{j+1}$ are compatible, and

  \item $t_{j}(\ybf_{j})$ and $t_{j+1}(\ybf_{j+1})$ coincide on
    $[\ybf_{j}]\cap [\ybf_{j+1}]$.

\end{itemize}
Let $E_i$ be the induced prefix-order on $T_i$.  We call
$(t_{n}(\ybf_{n}),\Phi_{n})$ the \emph{tail} of $\sigma_{n}$.  It
remains to define the functions $\text{bag}_{i}$ and $\text{Bag}_{i}$.
We give an inductive definition with the aim to achieve the following:
for all $\sigma_{n}\in T_{i}$ of the form above the $\Xi$-type
$t_{n}(\ybf_{n})$ is satisfied in $\Amf_{i}$ under a canonical
assignment $v_{\sigma_{n}}$ into the set $\text{bag}_{i}(\sigma_{n})$.
For the construction, it is important to note that we have $\neg
(x=y)\in t$ for any two distinct free variables $x,y$ in any
$\Xi$-type $t$. Thus we can essentially use (copies of) the variables
$\ybf_n$ to define $\text{bag}_{i}(\sigma_n)$.


For the inductive definition, start by setting
$\text{bag}_{i}(\sigma_{0})=\emptyset$ and $v_{\sigma_0}=\emptyset$ for
$\sigma_{0}= (\hat{t}_{i},\Phi_{0})$. In the inductive step,
assume that $\text{bag}_{i}$, $v_{\sigma_{n-1}}$, and $\text{Bag}_{i}$  have been defined on $\sigma_{n-1}$, where
$$
\sigma_{n-1}= (t_{0}(\ybf_{0}),\Phi_{0}),\ldots,(t_{n-1}(\ybf_{n-1}),\Phi_{n-1}).
$$
Then $\text{bag}_{i}(\sigma_{n})$ contains 
\begin{itemize}

  \item fresh copies $y'$ of the variables $y\in [\ybf_{n}]\setminus
    [\ybf_{n-1}]$ and 

  \item $v_{\sigma_{n-1}}(y)$ for every $y\in[\ybf_{n}]\cap [\ybf_{n-1}]$,
  
\end{itemize}
and $v_{\sigma_{n}}(y)$ is defined as the copy $y'$ of $y$ for $y\in
[\ybf_{n}]\setminus [\ybf_{n-1}]$ and by setting $v_{\sigma_{n}}(y):=
v_{\sigma_{n-1}}(y)$ for $y\in[\ybf_{n}]\cap [\ybf_{n-1}]$. Finally,
we define $\text{Bag}_{i}(\sigma_{n})$ by interpreting 
any relation symbol $R$ in such a way that the atomic formulas in
$t_{n}(\ybf_{n})$ are satisfied under $v_{\sigma_{n}}$, that is,
such that $\text{Bag}_{i}(\sigma_{n})$ satisfies 
$R(v_{\sigma_n}(\ybf))$ iff $R(\ybf) \in t_n(\ybf_n)$.

Let $\Amf_{i}$ be the union of all $\text{Bag}_{i}(t)$, $t\in T_{i}$. 
It is easy to see that $(T_{i},E_i,\text{bag}_{i})$ is a tree decomposition 
of $\Amf_{i}$. In fact, in the inductive step above, $t_n(\ybf_n)$ and
$t_{n-1}(\ybf_{n-1})$ coincide on $[\ybf_n]\cap [\ybf_{n-1}]$. Thus,
the interpretation of any relation symbol $R$ coincides 
on the intersection of $\text{bag}_{i}(\sigma_n)$ and $\text{bag}_{i}(\sigma_{n-1})$. We proceed to show that the GF($\tau$)-bisimulation mentioned in Lemma~\ref{lem:construct} indeed
exists. To this end, we prove the following auxiliary claim. 
We prove the following in the appendix:

\medskip
\noindent
{\bf Claim 1.}
{\it For all $i,j$ with $1\leq i, j\leq m$,
  we have: \begin{enumerate}
  \item For every $\sigma\in T_{i}$ with
    $\text{tail}(\sigma)=(t(\ybf),\Phi)$, we have $\Amf_{i}\models
    t(v_{\sigma}(\ybf))$;

  \item Let $H_{i,j}$ be the set of all mappings
    $p_{\sigma,\sigma',\zbf}$, where 
    \begin{itemize}

      \item $\sigma\in T_i,\sigma'\in T_j$, 
	$\text{tail}(\sigma)=(t(\ybf),\Phi)$, and
	$\text{tail}(\sigma')=(t'(\ybf'),\Phi)$; 

      \item $\zbf$ is a tuple with $[\zbf] \subseteq [\ybf]\cap
	[\ybf']$ and $v_{\sigma}(\zbf)$ is $\tau$-guarded in $\Amf_i$
	$($or, equivalently, $v_{\sigma'}(\zbf)$ is $\tau$-guarded in $\Amf_j$$)$;

      \item $p_{\sigma,\sigma',\zbf}: v_{\sigma}(\zbf) \mapsto
	v_{\sigma'}(\zbf)$.  

    \end{itemize} 
    Then $H_{i,j}$ is a GF($\tau$)-bisimulation between $\Amf_{i}$
    and $\Amf_{j}$. 
\end{enumerate}
}
To complete the proof of Lemma~\ref{lem:construct}, assume that $\hat{t}_{i}
\subseteq t_{i}(\xbf_{i})$ for $i=1,2$ (the case $\hat{t}_{i}
\subseteq t_{1}(\xbf_{1}), t_{2}(\xbf_{2})$ for some $i$ is similar). Take $\rho_{i}=(\hat{t}_{i},\widehat{\Psi})\cdot(t_{i}(\xbf_{i}),\Psi)\in T_{i}$,
for $i=1,2$. Consider the tuples $\abf_{i}:= v_{\rho_{i}}(\xbf_{i})$.
By Claim~1, $\Amf_{i}\models t_{i}(\abf_{i})$.
Also by Claim 1, for any tuple $\zbf$ with $[\zbf] \subseteq [\xbf_{1}]\cap [\xbf_{2}]$ and 
such that $v_{\rho_{1}}(\zbf)$ is $\tau$-guarded in $\Amf_1$ or
$\Amf_2$,
we have $p_{\rho_{1},\rho_{2},\zbf}: v_{\rho_{1}}(\zbf)\mapsto
v_{\rho_{2}}(\zbf)\in H_{1,2}$. As any two $p_{\rho_{1},\rho_{2},\zbf}$ coincide on the intersection of their domains, we have  
$\Amf_{1},v_{\rho_{1}}([\xbf_{1}]\cap [\xbf_{2}]) \sim_{\text{GF},\tau} \Amf_{2},v_{\rho_{2}}([\xbf_{1}]\cap [\xbf_{2}])$, as required.
\end{proof}
We next show how to read off an existentially saturated set of mosaics from
jointly consistent structures, as illustrated in Example~\ref{ex:internal}.
We make sure that all mosaics except a single mosaic $\Psi$ use only
width$(\Xi)$ many free variables and that also in $\Psi$ only at most two types use more variables.
\begin{lemma}\label{lem:readoff}
  Let $\Amf_{1},\abf_{1}$ and $\Amf_{2},\abf_{2}$ be pointed
  structures with $\abf_{1}$ and $\abf_{2}$ tuples with pairwise
  distinct elements of length $m\leq \text{fv}(\Xi)$ and let $\tau$ be
  a signature. Consider assignments $\xbf_{0} \mapsto \abf_{i}$ with
  $[\xbf_{0}]\subseteq \{x_{0},\ldots,x_{2n}\}$. If $\Amf_{1},\abf_{1}
  \sim_{\text{GF},\tau} \Amf_{2},\abf_{2}$, then there exists an existentially saturated set $\mathcal{M}$ of $\tau$-mosaics and some $\Psi\in \Mmc$ such that
  \begin{itemize} 

    \item all $\Phi\in \mathcal{M}$ with $\Phi\neq \Psi$ use at most
      $\text{width}(\Xi)$ many free variables; 


   \item there exist types $t_1(\xbf_0),t_2(\xbf_0)\in \Psi$ such that
      $t_i(\xbf_0)=\text{tp}(\Amf_{i},\xbf_{0} \mapsto \abf_{i})$ for
      $i=1,2$ and
      all types $t(\ybf)\in \Psi\setminus\{t_1(\xbf_0),t_2(\xbf_0)\}$
      use at most $\text{width}(\Xi)$ free variables among
      $[\xbf_0]$.

  \end{itemize}
\end{lemma} 
\begin{proof} Assume w.l.o.g. that $\Amf_{1}$ and
  $\Amf_{2}$ are disjoint. For any tuples $\bbf_{1}$ in $\Amf_{i}$ and
  $\bbf_{2}$ in $\Amf_{j}$ with $i,j\in \{1,2\}$, we use
  $\text{tp}(\xbf_{1}\mapsto \bbf_{1})$ to denote
  $\text{tp}(\Amf_{i},\xbf_{1} \mapsto \bbf_{1})$ and we write
  $\bbf_{1}\sim_{\text{GF},\tau} \bbf_{2}$ if
  $\Amf_{i},\bbf_{1}\sim_{\text{GF},\tau} \Amf_{j},\bbf_{2}$. For any tuple $\abf$ of distinct elements in either $\Amf_{1}$ or $\Amf_{2}$, and any tuple $\xbf$ with $[\xbf]\subseteq \{x_{1},\ldots,x_{2n}\}$ such that $v: \xbf \mapsto \abf$ is a
  bijection, let $\Phi_{\abf,\xbf}$ be the set of all types
  $\text{tp}(v':\xbf_{|Y} \mapsto \bbf)$ with $Y\subseteq [\xbf]$
  and $\bbf$ in either $\Amf_{1}$ or $\Amf_{2}$ such that
  $v(\xbf_{|Y}) \sim_{\text{GF},\tau} v'(\xbf_{|Y})$.
    
  Let $\mathcal{M}$ contain all such $\Phi_{\abf,\xbf}$ with $\abf$ of length at most $\text{width}(\Xi)$ and $\xbf$ from
    $\{x_{1},\ldots,x_{2n}\}$. Moveover, if $m>\text{width}(\Xi)$, then
    add $\widehat{\Phi}_{\abf_{1},\xbf_{0}}$ to $\mathcal{M}$,
    where $\widehat{\Phi}_{\abf_{1},\xbf_{0}}$ is obtained from 
    $\Phi_{\abf_{1},\xbf_{0}}$ by removing all $t$ distinct from 
    $t_{1}(\xbf_{0})$ and $t_{2}(\xbf_{0})$ using more than $\text{width}(\Xi)$ many free variables. 
    
    We show in the appendix that $\mathcal{M}$ is as required.
\end{proof}
 
It follows from Lemmas~\ref{lem:construct} and~\ref{lem:readoff} that
the following two conditions are equivalent, where 
$\mathcal{M}'$ is the maximal existentially saturated set of
$\tau$-mosaics using at most $\text{width}(\Xi)$ free
variables.
\begin{enumerate}[label=\arabic*.]

  \item $\varphi(\xbf_0),\psi(\xbf_0)$ are jointly GF$(\tau)$-consistent;

  \item There exists a $\tau$-mosaic $\Psi$ and $\Xi$-types $t_1(\xbf),t_2(\xbf)\in \Psi$ such that $\mathcal{M}=\{\Psi\}\cup \Mmc'$ 
  is existentially saturated and:
    \begin{enumerate}

      \item $t_1(\xbf),t_2(\xbf)$ have $\text{fv}(\Xi)$
      many free variables
	and one can replace the variables in $[\xbf_{0}]$ by
	variables in $\xbf$ such that $\varphi'\in t_{1}(\xbf)$,
	$\psi'\in t_{2}(\xbf)$ for the resulting formulas
	$\varphi',\psi'$;

      \item all $\Xi$-types $t(\ybf)\in
	\Psi\setminus\{t_1(\xbf),t_2(\xbf)\}$ use at most
	$\text{width}(\Xi)$ free variables among $[\xbf]$;

    \end{enumerate}

\end{enumerate}
Hence, it suffices to provide an algorithm deciding Condition~2.
\begin{lemma}\label{lem:complexity}
  On input $\varphi(\xbf_0),\psi(\xbf_0)$, Condition~2 can be 
%
  decided in time triple exponential in the size of
  $\varphi(\xbf_{0}),\psi(\xbf_{0})$ in general, and double
  exponential in the size of $\varphi(\xbf_{0}),\psi(\xbf_{0})$ if
  $\text{width}(\Xi)$ is bounded by a constant.
\end{lemma}
%
\begin{proof}
  %
%
%
%
  %
%
  First determine $\mathcal{M}'$ by exhaustively removing $\tau$-mosaics that are not existentially saturated from the list of all $\tau$-mosaics with at most width$(\Xi)$ free variables. It can be verified that the fixpoint is existentially saturated.
  Next we proceed as follows: for every pair $t_1(\xbf),t_2(\xbf)$ of $\Xi$-types  
  that satisfies Condition~2(a) enumerate all $\tau$-mosaics $\Psi$
  satisfying Condition~2(b), that is, $t_1(\xbf),t_2(\xbf)\in
  \Psi$ and all types in $\Psi$ except $t_1(\xbf),t_2(\xbf)$ use
  at most $\text{width}(\Xi)$ free variables among $[\xbf]$. Accept
  if at least one $\{\Psi\} \cup \mathcal{M}'$ is existentially saturated.
  Reject otherwise.
  
  Correctness of the algorithm is straightforward, so it remains to
  analyze its run time. For this purpose, let $r$ be the number of
  subformulas (of formulas) in $\Xi$ and $\ell\geq 0$.  Observe that a
  subformula with $\ell$ free variables has at most $(2n)^{\ell}$
  instantiations with variables from $x_1,\ldots,x_{2n}$. Since for
  every such instantiated formula either the formula itself or its
  negation is contained in any type, there are at most $2^{r
    (2n)^{\ell}}$ many types with $\ell$ free variables. Thus, there
    are only double exponentially many choices for
    $t_{1}(\xbf),t_2(\xbf)$ and $\Psi$. Moreover, the set of all 
    $\tau$-mosaics with at most width$(\Xi)$ free variables
    is of size triple exponential in the size of
    $\varphi(\xbf_{0}),\psi(\xbf_{0})$ in general, and double
    exponential in the size of $\varphi(\xbf_{0}),\psi(\xbf_{0})$ if
    $\text{width}(\Xi)$ is bounded by a constant. The upper bounds now follow from the observation that checking whether
    some $\Phi$ is existentially saturated in some set $\mathcal{M}_{0}$ of mosaics can be done in time polynomial in the size of $\Mmc_{0}$. 
%
\end{proof}

From the equivalence of Conditions~1 and~2, and
Lemma~\ref{lem:complexity} we finally obtain that joint
GF$(\tau)$-consistency is in \ThreeExpTime in general, and in
\TwoExpTime if the arity of relation symbols is bounded by a constant.
%
%

\subsection{Lower Bounds}
We reduce the word problem for languages recognized
by exponentially
and double exponentially space bounded alternating Turing machines, respectively. An \emph{alternating Turing machine (ATM)} is a tuple
$M=(Q,\Theta,\Gamma,q_0,\Delta)$ where $Q=Q_{\exists}\uplus
Q_{\forall}$ 
is the set of states that consists of \emph{existential states}
in~$Q_{\exists}$ and \emph{universal states} in~$Q_{\forall}$.
Further, $\Theta$ is the input alphabet and $\Gamma$ is the tape
alphabet that contains a \emph{blank symbol} $\Box \notin \Theta$,
$q_0\in Q_{\forall}$ is the \emph{starting state}, and the
\emph{transition relation} $\Delta$ is of the form
$\Delta\subseteq Q\times \Gamma\times Q\times \Gamma \times \{L,R\}.$
The set $\Delta(q,a):=\{(q',a',M)\mid (q,a,q',a',M)\in\Delta\}$ must
contain exactly two or zero elements for every $q\in Q$ and $a \in
\Gamma$. Moreover, the state $q'$ must be from $Q_\forall$ if $q \in
Q_\exists$ and from $Q_\exists$ otherwise, that is, existential and
universal states alternate. 

We use the following (slightly
non-standard) acceptance condition. 
A \emph{configuration} of an ATM is a word $wqw'$ with
\mbox{$w,w'\in\Gamma^*$} and $q\in Q$. 
We say that $wqw'$ is \emph{existential} if~$q$ is, and likewise for
\emph{universal}.  \emph{Successor configurations} are defined in the
usual way.  Note that every configuration has either zero or two successor
configurations.  A \emph{computation tree} of an ATM $M$ on
input $w$ is an infinite tree whose nodes are labeled
with configurations of $M$ such that
\begin{itemize}

  \item the root is labeled with the initial configuration
    $q_0w$;

  \item if a node is labeled with an existential configuration 
    $wqw'$, then it has a single successor and this successor is labeled 
    with a successor configuration of $wqw'$;

  \item if a node is labeled with a universal configuration
    $wqw'$, then it has two successors and these successors are
    labeled with the two successor configurations of~$wqw'$.

\end{itemize}
An ATM $M$ \emph{accepts} an input $w$ if there is a computation tree
of $M$ on $w$. Note that, starting from
the standard ATM acceptance condition defined via accepting states, this can be achieved by assuming that
exponentially (resp., double exponentially) space bounded ATMs
terminate on every input and then
modifying them to enter an infinite loop from the accepting state.
It is well-known that there are $2^n$-space
bounded and $2^{2^n}$-space bounded ATMs for which the recognized
language is
\TwoExpTime-hard and \ThreeExpTime-hard,
respectively~\cite{chandraAlternation1981}. 

\subsubsection{Bounded Arity}
%
For didactic reasons, we start with showing \TwoExpTime-hardness for the
bounded arity case. Let $M$ be a $2^n$-space bounded ATM and $w$ an input.
The idea of the reduction is as follows. We set 
\begin{align*}
%
%
\tau & = \{R,S,X,Z,B_\forall,B_\exists^1,B_\exists^2\}\cup \{A_\sigma\mid \sigma\in \Gamma\cup
(Q\times \Gamma)\},
\end{align*}
where $R,S$ are binary relation symbols, and the remaining symbols are
unary.  
We aim to construct $\varphi$ such that $M$ accepts $w$ iff
$\varphi\wedge A(x)$ and $\varphi\wedge \neg A(x)$ are jointly
GF$(\tau)$-consistent.  The sentence $\varphi$ is a conjunction of
several GF-sentences, which are, except for one, also 
FO$^2$-sentences. The first conjunct, $\varphi_0$ below, is this
exception and enforces that
every element satisfying $A$ is involved in a three-element $R$-loop
(similar to Example~\ref{ex:joint}):  
\begin{align*}
  \varphi_0 & = \forall x\big(A(x)\rightarrow \exists yz(G(x,y,z)\wedge
  X(x)\wedge \neg X(y)\wedge{} \\ 
  & \hspace{2.5cm} 
  \neg X(z) \wedge R(x,y)\wedge R(y,z)\wedge R(z,x)\big)
\end{align*}
%
Now, if $\varphi\wedge A(x)$ and $\varphi\wedge \neg A(x)$ are jointly
GF$(\tau)$-consistent,
 there exist models \Amf and \Bmf of $\varphi$ and elements $a,b$ such that $a\in A^\Amf$, $b\notin A^\Bmf$, and
$\Amf,a\sim_{\text{GF},\tau} \Bmf,b$. If the latter holds, then from $a\in A^\Amf$ and
$\varphi_0$ it follows that $b$ has an infinite outgoing path $\rho$
along $R$ on which every third element satisfies $X$ and is guarded
$\tau$-bisimilar to $a$.  Let us call these elements the
\emph{$X$-elements}.  As guarded bisimilarity is an equivalence
relation, all $X$-elements are actually guarded $\tau$-bisimilar.
The other conjuncts of $\varphi$
will enforce that along the $X$-elements on $\rho$, a
counter counts modulo $2^{n}$ using relation symbols not in $\tau$.
Moreover, in every $X$-element of $\rho$ starts an infinite tree along symbol $S$ that is
supposed to mimick the computation tree of $M$. Along this tree, two
counters are maintained:
\begin{itemize}
	
  \item one counter starting at $0$ and counting modulo $2^n$ to
    divide the tree in subpaths of length $2^n$; each such path of
    length $2^n$ represents a configuration;

  \item another counter starting at the value of the counter along
    $\rho$ and also counting modulo $2^{n}$.

\end{itemize}
To link successive configurations we use the fact that all
$X$-elements on $\rho$ are guarded $\tau$-bisimilar and thus each
$X$-element is the starting point of trees along $S$ with
identical $\tau$-decorations. As on the $m$th such tree the second
counter starts at all nodes at distances $k\cdot 2^{n}-m$, for all
$k\geq 1$, we are in the position to coordinate all positions at all
successive configurations.

In detail, let $w=a_0,\dots,a_{n-1}$ be an input to $M$ of length $n$.
We will be using unary symbols $A_i, U_i, V_i$, $1\leq i\leq n$ to
represent the aforementioned binary counters; we will refer to them
with $A$-counter, $U$-counter, and $V$-counter, respectively. 

The sentences below enforce that the $A$-counter along the $R$-path
$\rho$ is incremented (precisely) at every $X$-element. To
avoid that the counter is started at $a$ (which would lead to a
contradiction), we use an additional symbol $I\notin\tau$ that is
satisfied along the entire path and acts as an additional
guard:\footnote{Throughout, we assume that $\wedge$ has higher
  precedence than $\rightarrow$. Moreover, some formulas are not syntactically guarded but can
easily be rewritten.}
\begin{align*}
  \forall xy \big(R(x,y) & \rightarrow (\neg A(x)\wedge X(x) \rightarrow I(x)
  \big) \\
 \forall xy \big(R(x,y) & \rightarrow (I(x)\leftrightarrow
 I(y))\big)\\
 \forall xy \big(R(x,y) \wedge I(x) \wedge \neg X(y) & \rightarrow \text{Eq}(x,y))\big) \\
 \forall xy \big(R(x,y) \wedge I(x) \wedge
 X(y) & \rightarrow \text{Succ}(x,y))\big)
%
%
\end{align*}
Here, atoms $\text{Eq}(x,y)$ and
$\text{Succ}(x,y)$ are abbreviations for formulas that express that
the $A$-counter value at $x$ equals (resp., is the predecessor
of) the $A$-counter value at $y$, that is: 
\begin{align*}
  \text{Eq}(x,y) & =  \textstyle\bigwedge_{i}A_i(x)\leftrightarrow A_i(y) \\
  \text{Succ}(x,y) & =  \textstyle\bigvee_{i}\Big( A_i(y)\wedge \neg A_i(x)
  \wedge
  \bigwedge_{j<i} (\neg A_j(y) \wedge A_j(x))\\
& \hspace{1.5cm} 
  \textstyle\wedge\bigwedge_{j>i} (A_j(y) \leftrightarrow A_j(x)) \Big)
\end{align*}
Now, we start a tree along $S$ from all $X$-elements on the infinite
$R$-path. Along the path, we maintain the $U$- and $V$-counter, which
are initialized to $0$ and the value of the $A$-counter, respectively:
\begin{align*}
  \forall x \exists yS(x,y) &\\
  \forall xy\big(S(x,y) & \rightarrow (I(x)\leftrightarrow I(y)\big)\\
  \forall x\big(I(x)\wedge X(x) & \rightarrow \text{Min}_U(x)\big)\\
  \forall x \big(I(x)\wedge X(x) & \rightarrow \textstyle\bigwedge_{i}
  (V_i(x)\leftrightarrow A_i(x))\big)
\end{align*}
Here, $\text{Min}_U(x)$ is an abbreviation for the formula that
expresses that the $U$-counter is $0$ at~$x$; we use similar
abbreviations such as $\text{Max}_V(x)$ below.  The $U$ and
$V$-counters are incremented along $S$ analogously to how the
$A$-counter is incremented along $R$, but on every $S$-step; we omit
details. Configurations of $M$ are represented between two consecutive
elements having $U$-counter value $0$. We next enforce the structure
of the computation tree (recall that $q_0\in Q_{\forall}$). 
\begin{align*}
  \forall x\big(I(x)\wedge X(x) & \rightarrow B_\forall(x)\big) \\
  \forall xy\big(S(x,y) \wedge I(x)\wedge \neg\text{Max}_U(x) & \rightarrow
  (B_\forall(x) \leftrightarrow B_{\forall}(y))\big) \\
  \forall xy\big(S(x,y) \wedge I(x) \wedge \neg\text{Max}_U(x)
  & \rightarrow\textstyle\bigwedge_{j=1}^2
  (B_\exists^j(x) \leftrightarrow B_{\exists}^j(y))\big)\\
   \forall xy\big(S(x,y) \wedge I(x) \wedge \text{Max}_U(x) & \rightarrow
   (B_\forall(x) \leftrightarrow \neg B_{\forall}(y))\big)\\
  \forall x\big( I(x)\wedge \text{Max}_U(x) & \rightarrow \exists y(S(x,y)\wedge
  Z(y))\wedge{} \\
  & \phantom{ {}\rightarrow{} } \exists y(S(x,y)\wedge \neg Z(y)) \big) \\
  \forall x\big( I(x)\wedge \neg B_{\forall}(x) & \rightarrow (B_\exists^1(x)
  \leftrightarrow \neg B_{\exists}^2(x))\big)
\end{align*}
These sentences enforce that all nodes which represent a
configuration satisfy exactly one of $B_\forall, B_{\exists}^1,
B_{\exists}^2$, indicating the kind of configuration and, if
existential, also a choice of the transition function, indicated in
the superscript of $B_{\exists}^j$. The symbol
$Z\in \tau$ enforces the branching. 

We next set the initial configuration, for input
$w=a_0,\dots,a_{n-1}$.  Below, we use $\forall y ( S^i(x,y)\rightarrow
\psi(y))$ to abbreviate the GF-formula that enforces $\psi$ at all
elements $y$ that are reachable in $i$ steps via $S$ from $x$.
\begin{align*}
  %
  & \forall x (I(x)\wedge X(x)\rightarrow A_{q_0,a}(x)) \\
  & \forall x\big(I(x)\wedge X(x) \rightarrow \forall
  y(S^k(x,y)\rightarrow A_{a_k}(y))\big),\quad  0<k<n \\
  & \forall x\big(I(x)\wedge X(x) \rightarrow \forall
  y(S^{n}(x,y)\rightarrow \text{Blank}(y))\big)\\
  & \forall x (\text{Blank}(x) \rightarrow A_{\Box}(x)) \\
  & \forall x (\text{Blank}(x)\wedge \neg\text{Max}_U(x)\rightarrow
  \forall y(S(x,y)\rightarrow \text{Blank}(y)))
\end{align*}
We next coordinate consecutive configurations, focusing on cells that
are not at the border of a configuration; these corner cases can be
dealt with accordingly.  To this end, we associate with $M$ functions
$f_i$, $i\in\{1,2\}$ that map the content of three consecutive cells
of a configuration to the content of the middle cell in the $i$-th
successor configuration (assuming an arbitrary order on the
$\Delta(q,a)$). Moreover, for each triple
$(\sigma_1,\sigma_2,\sigma_3)\in (\Gamma\cup(Q\times\Gamma))^3$, we
fix a GF-formula $\psi_{\sigma_1,\sigma_2,\sigma_3}(x)$ that is
satisfied at an element $a$ of the computation tree iff $a$ is labeled
with $A_{\sigma_2}$, $a$ has an $S$-predecessor labeled with
$A_{\sigma_1}$, and $a$ has an $S$-successor labeled with
$A_{\sigma_3}$. Now, in each configuration, we synchronize elements
with $V$-counter $0$, by including for every
$\vec{\sigma}=(\sigma_1,\sigma_2,\sigma_3)$ and $i\in \{1,2\}$ the
following sentences:
\begin{align*}
  \forall x\big(& I(x)\wedge \text{Min}_V(x)\wedge
  \neg\text{Min}_U(x)\wedge\neg \text{Max}_U(x) \wedge B_{\forall}(x)
  \rightarrow \\
  &\ \ (\psi_{\vec\sigma}(x) \rightarrow A_{f_1(\vec\sigma)}^1(x) \wedge
  A_{f_2(\vec\sigma)}^2(x)) \big) \\
  \forall x\big(& I(x)\wedge \text{Min}_V(x)\wedge
  \neg\text{Min}_U(x)\wedge\neg \text{Max}_U(x) \wedge
  B_{\exists}^i(x) \rightarrow \\
  &\ \ (\psi_{\vec\sigma}(x) \rightarrow A_{f_i(\vec\sigma)}^i(x)) \big)
\end{align*}
The unary symbols $A^i_{\sigma}$ are used as markers (not in $\tau$)
and are propagated along $S$ for $2^n$ steps, exploiting the
$V$-counter.  The superscript $i\in\{1,2\}$ determines the successor
configuration that the symbol is referring to. After crossing the end
of a configuration, the symbol $\sigma$ is propagated using further
unary symbols $A_{\sigma}'$ (the superscript is not needed anymore
because the branching happens at the end of the configuration, based
on $Z$):
\begin{align*}
  %
  %
  %
  \forall x\big(\neg \text{Max}_U(x)\wedge A_{\sigma}^i(x) &
  \rightarrow \forall y(S(x,y)\rightarrow A_{\sigma}^i(y)) \big) \\
  \forall x\big(\text{Max}_U(x) \wedge B_{\forall}(x) \wedge
  A^1_{\sigma}(x) & \rightarrow{} \\
  & \hspace{-1cm} \forall y(S(x,y) \rightarrow(Z(y)
  \rightarrow A'_{\sigma}(y))) \big)\\
  \forall x\big(\text{Max}_U(x) \wedge B_{\forall}(x)\wedge
  A^2_{\sigma}(x) & \rightarrow{} \\
  & \hspace{-1cm}
  \forall y(S(x,y)\rightarrow(\neg Z(y)
  \rightarrow A'_{\sigma}(y))) \big)\\
  \forall x\big(\text{Max}_U(x) \wedge B_{\exists}^i(x)\wedge
  A^i_{\sigma}(x) & \rightarrow \forall y(S(x,y)\rightarrow
  A'_{\sigma}(y))\big) \\ 
  \forall x\big(\neg \text{Max}_V(x)\wedge A'_\sigma(x) &\rightarrow
  \forall y(S(x,y)\rightarrow A'_\sigma(x))\big) \\
  \forall x\big(\text{Max}_V(x)\wedge A'_\sigma(x) &\rightarrow
  \forall y(S(x,y)\rightarrow A_\sigma(x))\big)
\end{align*}
For those $(q,a)$ with $\Delta(q,a)=\emptyset$, we add the sentence
\[\forall x\; \neg A_{q,a}(x) \]
%
to ensure that such halting states are never reached. Correctness of
the reduction is established in the appendix.

\begin{restatable}{lemma}{lemcorrectgf}\label{lem:correct2exp}
  $M$ accepts the input $w$ iff there exists models $\Amf,\Bmf$ of $\varphi$
  and elements $a\in A^\Amf$, $b\notin A^\Bmf$ 
  such that $\Amf,a\sim_{\text{GF},\tau}\Bmf,b$.
\end{restatable}

\subsubsection{General Case}
We reduce the word problem of $2^{2^n}$-space bounded ATMs using  
the very same idea as in the previous section.
However, we need double exponential counters instead of the single
exponential counters for $A,U,V$ above. These counters are encoded in a way
similar to the \TwoExpTime-hardness proof for satisfiability in the
guarded fragment~\cite{DBLP:journals/jsyml/Gradel99}. The mentioned
encoding is based on \emph{pairs of elements}, so we ``lift'' the
above reduction to pairs of elements and consequently double the arity of
all involved symbols. More precisely, we set 
\begin{align*}
%
%
\tau & = \{R,S,X,Z,B_\forall,B_\exists^1,B_\exists^2\}\cup \{A_\sigma\mid \sigma\in \Gamma\cup
(Q\times \Gamma)\},
\end{align*}
where $R,S$ are 4-ary relation symbols, and the remaining symbols are
binary. The sentence $\varphi$ is a conjunction of
several sentences. The first conjunct, $\varphi_0$ below, enforces that
every pair of elements satisfying $A$ is involved in a three-element
$R$-loop as follows:\footnote{We omit commas and/or parentheses when no confusion
can arise.}
\begin{align*}
  \varphi_0& =\forall xx'\big(Axx'\rightarrow \exists
  yy'zz'(Gxx'yy'zz'\wedge
  Xxx' \wedge{} \\
  & \neg Xyy'\wedge\neg Xzz'\wedge Rxx'yy'\wedge
  Ryy'zz'\wedge Rzz'xx')\big) 
\end{align*}
As above, we aim to construct $\varphi$ such that $M$ accepts $w$ iff
there exist models \Amf and \Bmf of $\varphi$ and pairs $\abf,\bbf$ such that $\abf\in A^\Amf$, $\bbf\notin A^\Bmf$, and
$\Amf,\abf\sim_{\text{GF},\tau} \Bmf,\bbf$. If the latter holds 
then from $\abf\in A^\Amf$ and
$\varphi_0$ it follows that $\bbf$ has an infinite outgoing ``path'' $\rho$
along $R$ on which every third pair of elements satisfies $X$ and is guarded
$\tau$-bisimilar to $\abf$. Let us call these pairs the
\emph{$X$-pairs}. Observe that all $X$-pairs are guarded
$\tau$-bisimilar.

The main difference to the reduction above is the realization of the
counters, so we will concentrate on this and leave the
(straightforward) remainder of the proof to the reader. For realizing
the $A$-counter, we use an $n$-ary relation symbol $D$ and associate
a counter to every pair of elements $(a,a')$ as follows. We assume the
order $a<a'$ which induces an order $<$ on tuples $\abf\in
\{a,a'\}^n$. Thus, every tuple $\abf\in \{a,a'\}^n$ corresponds to a number
$r(\abf)<2^n$, the rank of $\abf$ according to $<$. Now the sequence
of truth values on all these tuples in $D$ can be viewed as the
binary representation of a number $<2^{2^n}$.

The $A$-counter along the $R$-path $\rho$ is enforced by the
following sentences:
\begin{align*}
  \forall xx'yy'\big(Rxx'yy'
  & \rightarrow (\neg Axx' \wedge Xxx'\rightarrow  Ixx') \big) \\
  \forall xx'yy' \big(Rxx'yy' & \rightarrow (Ixx'\leftrightarrow
  Iyy')\big)\\
  %
%
  \forall xx'yy' \big(Rxx'yy' \wedge Ixx' 
  & \rightarrow (\neg Xyy'\rightarrow \text{Eq}(xx'yy'))\big) \\
  \forall xx'yy' \big(Rxx'yy' \wedge Ixx' 
  & \rightarrow (Xyy'\rightarrow \text{Succ}(xx'yy'))\big)
\end{align*}
Again, the $I$ acts as an additional guard that disables the counting
at $\abf$.  It remains to define the formulas $\text{Eq}(xx'yy')$ and
$\text{Succ}(xx'yy')$. We show in the appendix that we can axiomatize
a $(4n+4)$-ary predicate $E$ such that, for pairs $a,a'$ and
$b,b'$ where $b,b'$ represents a successor node of $a,a'$, and for
$\abf,\abf'\in\{a,a'\}^n$ and $\bbf,\bbf'\in \{b,b'\}^n$, we have
\begin{equation*}
  E(\abf\abf'aa'\bbf\bbf'bb')\text{\quad iff\quad }
  r(\abf)=r(\bbf)\text{ and }r(\abf')=r(\bbf'). 
\end{equation*}
Then the formulas $\text{Eq}$ and $\text{Succ}$ can be defined as
follows:
\begin{align*}
  \text{Eq}(xx'yy') & =
  \forall\xbf\ybf\xbf'\ybf'\big(E\xbf\xbf'xx'\ybf\ybf'yy'
  \rightarrow(D\xbf\leftrightarrow D\ybf)\big)\\
  \text{Succ}(xx'yy') &= \exists \xbf\ybf \big(E\xbf\xbf xx'\ybf\ybf
  yy' \wedge \neg D\xbf\wedge D\ybf\\
  & \hspace{-1cm} \wedge \forall
  \xbf'\ybf'\big(E\xbf\xbf'xx'\ybf\ybf'yy'\rightarrow \\
  & (\text{less}(\xbf'\xbf xx')\rightarrow D\xbf'\wedge \neg
  D\ybf' )\wedge {} \\
  & 
  (\text{less}(\xbf\xbf'xx')\rightarrow (D\xbf'\leftrightarrow
  D\ybf' ))\big) \big)
\end{align*}
where, for $\xbf=x_0\ldots x_{n-1}$ and $\xbf'=x_0'\ldots x_{n-1}'$,
we have 
\begin{align*}
  \text{less}(\xbf\xbf'xx')=\textstyle\bigvee_{i<n}\big(x_i'=x'\wedge
  x_i=x\wedge \bigwedge_{j>i} x_j=x_j' \big). 
\end{align*}
Thus, $\text{less}(\xbf\xbf'xx')$ compares the positions of
$\xbf$ and $\xbf'$ according to the order $x<x'$. Moreover,
$\text{Eq}(xx'yy')$ is true iff the counters stipulated by 
$x,x'$ and $y,y'$ have precisely the same bits set. Finally, 
$\text{Succ}(xx'yy')$ asserts the existence of a position $k$ such that
\textit{(i)} in the
counter stipulated by $x,x'$ bit $k$ is set to $0$ while in the
counter stipulated by $y,y'$ bit $k$ is set to $1$, \textit{(ii)} on
all positions $k'$ less than $k$, the bits in the former counter
are $1$ while the bits in the latter are $0$, and
\textit{(iii)} on all positions $k'$ greater than $k$ the counters
agree on their bits.

Having the adapted counters available, the proof then proceeds along
the lines of the proof given for the bounded arity case, always
replacing single elements/variables with pairs of elements/variables
as exemplified above. 

\section{Deciding Joint FO$^{2}(\tau)$-consistency}\label{sec:fo2}

We prove Theorem~\ref{lem:lower-bounds} (ii).  We proceed similarly to
the proof for GF by proving a N\TwoExpTime upper bound for joint
FO$^{2}(\tau)$-consistency and then applying
Lemma~\ref{lem:gf-char-interpolation} to obtain a \coNTwoExpTime upper
bound for FO$^{2}$-interpolant existence.  For the complexity lower
bound we consider joint FO$^{2}(\tau)$-consistency for an input
of the form given in Lemma~\ref{lem:gf-char-explicit}.

\subsection{Upper Bound}	
We show the N\TwoExpTime upper bound for joint FO$^{2}(\tau)$-consistency by proving that if two FO$^{2}$-formulas are jointly FO$^{2}(\tau)$-consistent, then there exist FO$^{2}(\tau)$-bisimilar models
satisfying the formulas of at most double exponential size:
\begin{theorem}
	If $\varphi(\xbf_{0}),\psi(\xbf_{0})$ are jointly
	FO$^{2}(\tau)$-consistent, then there are pointed models $\Bmf_{1},\bbf_{1}$ and $\Bmf_{2},\bbf_{2}$ of at most double exponential size such that $\Bmf_{1}\models \varphi(\bbf_{1})$, $\Bmf_{2}\models \psi(\bbf_{2})$ and $\Bmf_{1},\bbf_{1} \sim_{\text{FO}^2,\tau} \Bmf_{2},\bbf_{2}$.
\end{theorem}	
The remainder of this section is devoted to the proof. We first simplify the input
formulas. Generalizing~\cite{DBLP:journals/bsl/GradelKV97}, we show in
the appendix that one can assume w.l.o.g. that the input formulas only
use relation symbols of arity at most two. Then one can easily extend
the normal form for FO$^{2}$ sentences provided
in~\cite{DBLP:journals/bsl/GradelKV97} to the following normal form
for formulas: for any FO$^{2}$-formula $\chi(\xbf)$ only using
relation symbols of arity at most two one can construct in polynomial time an FO$^{2}$-formula $\chi'(\xbf)$ of the form
$$
\textstyle R_{0}(\xbf) \wedge \forall x \forall y \alpha \wedge \bigwedge_{i=1}^{m}\forall x \exists y \beta_{i}(x,y),
$$
where $R_{0}$ is a relation symbol and $\alpha$ and $\beta_{i}$ are quantifier-free such that all relations symbols in $\chi'(\xbf)$ have arity at most two and 
\begin{enumerate}
	\item $\chi'\models\chi$;
	\item every model of $\chi$ can be expanded to a model of $\chi'$.
\end{enumerate}
In what follows we can thus assume that the input formulas $\varphi(\xbf_{0}),\psi(\xbf_{0})$ are of this form. Let $\Xi=\{\varphi(\xbf_{0}),\psi(\xbf_{0})\}$. We use
$\text{cl}(\Xi)$ to denote the closure under single negation of the set of all subformulas of $\varphi$ and $\psi$ with at most the variable $x$ free and all formulas of the form $R(x)$ and $R(x,x)$ with $R$ a unary or, respectively, binary relation symbol in $\varphi,\psi$.
The \emph{1-type $t_{\Amf}(a)$ realized in a pointed structure $\Amf,a$} is defined as
$$
t_{\Amf}(a) := \{ \chi(x) \mid \Amf \models \chi(a), \chi\in \text{cl}(\Xi)\}
$$
A \emph{1-type} $t$ is any subset of $\text{cl}(\Xi)$ such that there exists  a pointed structure $\Amf,a$ with $t= t_{\Amf}(a)$.
A \emph{link-type} $l$ contains $x\not= y$ and for any binary relation symbol $R$ in $\Xi$ either $R(x,y)$ or $\neg R(x,y)$ and $R(y,x)$ or $\neg R(y,x)$. The \emph{link-type $l_{\Amf}(a,b)$ realized in a pointed structure $\Amf,a,b$ with $a\not=b$} contains $R(x,y)$ iff $\Amf \models R(a,b)$ and it contains $R(y,x)$ iff $\Amf\models R(b,a)$.

For a pair $(l,s)$ with $l$ a link-type and $s$ a 1-type we say that nodes $d,d'$ \emph{satisfy} $(l,s)$ in $\Amf$ if $l=l_{\Amf}(d,d')$ and $s=t_{\Amf}(d')$. 
%
%

Now assume that $\varphi$ and $\psi$ are jointly FO$^{2}(\tau)$-consistent. Then we find pointed models $\Amf_{1},\abf_{1}$ and $\Amf_{2},\abf_{2}$ satisfying $\varphi$ and $\psi$, respectively, such that $\Amf_{1},\abf_{1} \sim_{\text{FO}^2,\tau} \Amf_{2},\abf_{2}$. We extract from $\Amf_{1}$ and $\Amf_{2}$ new pointed models $\Bmf_{1},\bbf_{1}$ and $\Bmf_{2},\bbf_{2}$ which still witness joint FO$^{2}(\tau)$-consistency of $\varphi$ and $\psi$ but which are of at most double exponential size in $\varphi$ and $\psi$.
In what follows we assume that $\text{dom}(\Amf_{1}) \cap \text{dom}(\Amf_{2})=\emptyset$.
We write $d \sim_{\text{FO}^{2},\tau}e$ if there are $i,j\in \{1,2\}$ 
with 
$d\in \text{dom}(\Amf_{i})$, $e\in \text{dom}(\Amf_{j})$, and $\Amf_{i},d \sim_{\text{FO}^{2},\tau}\Amf_{j},e$. 

A \emph{mosaic} $m$ is a pair $(\Phi_{1},\Phi_{2})$ with $\Phi_{1},\Phi_{2}$  sets of 1-types. The \emph{mosaic} $m(d)=(\Phi_{1},\Phi_{2})$ \emph{generated by} $d\in \text{dom}(\Amf_{1})\cup \text{dom}(\Amf_{2})$ is defined by setting
$$
\Phi_{j}= \{ t_{\Amf_{j}}(e) \mid e\in \text{dom}(\Amf_{j}),
d \sim_{\text{FO}^{2},\tau}e\},
$$ 
for $j=1,2$. The set $\mathcal{M}$ of all mosaics generated in $\Amf_{1},\Amf_{2}$ is then defined as
$$
\mathcal{M}= \{m(d) \mid d \in \text{dom}(\Amf_{1})\cup \text{dom}(\Amf_{2})\}.
$$
Observe that since FO$^{2}(\tau)$-bisimulations are global, $\mathcal{M}=\{m(d) \mid d \in \text{dom}(\Amf_{i})\}$, for $i=1,2$. The set $\mathcal{K}\subseteq \mathcal{M}$ of \emph{king mosaics} is defined
as the set of all $m(d)\in \mathcal{M}$ such that for all $e$ with $m(d)=m(e)$ we have $d \sim_{\text{FO}^{2},\tau}e$. Let $\mathcal{C}=\mathcal{M}\setminus \mathcal{K}$ be the set of \emph{pawn mosaics}. If $m(d)=(\Phi_{1},\Phi_{2})$ is a king mosaic, then call any $t\in \Phi_{i}$ such that there exists exactly one $e$ with $t=t_{\Amf_{i}}(e)$ and $d \sim_{\text{FO}^{2},\tau} e$ an \emph{$i$-king} in $(\Phi_{1},\Phi_{2})$. Any $t\in \Phi_{i}$
that is not an $i$-king in $(\Phi_{1},\Phi_{2})$ is called an \emph{$i$-pawn} in $(\Phi_{1},\Phi_{2})$. All $t\in \Phi_{i}$ with $(\Phi_{1},\Phi_{2})$
a pawn mosaic are called \emph{$i$-pawns} in $(\Phi_{1},\Phi_{2})$. Note
that we generalize a few notions introduced in the single exponential size model property proof for FO$^{2}$ presented in~\cite{DBLP:journals/bsl/GradelKV97}. In that proof, 1-types that are realized exactly once in a model played a special roles and were called kings. Here we generalize that notion to king mosaics and kings within king mosaics.

\medskip
We are now in the position to define the domains of $\Bmf_{1},\Bmf_{2}$ as follows. Let $s$ be the size of the input $\varphi(\xbf_{0}),\psi(\xbf_{0})$.
Then the number of mosaics is bounded by
$m_{\varphi,\psi}=2^{2^{s+1}}$.
Let $k_{1}= 2^{4s}\times m_{\varphi,\psi}$ and let $k_{2}=2^{3s} \times k_{1}^{2}$. 

%

Take $k_{1}$ many copies $(t,1),(t,2),\ldots,(t,k_{1})$ of every 1-type $t$
and take $k_{2}$ many copies $(m,1),(m,2),\ldots,(m,k_{2})$ of every pawn mosaic $m$. Then the domain $\text{dom}(\Bmf_{i})$ of $\Bmf_{i}$ contains, for $i=1,2$: 
\begin{enumerate}
	\item \emph{new $i$-kings} $(t,m)$, for $m\in \mathcal{K}$ and $t$ an $i$-king in $m$;
	\item \emph{semi $i$-pawns} $((t,1),m),\ldots,((t,k_{1}),m)$ for $m \in \mathcal{K}$ and $t$ an $i$-pawn in $m$;
	\item \emph{full $i$-pawns} $((t,1),(m,j)),\ldots,((t,k_{1}),(m,j))$, for $m$ a pawn mosaic, $t$ an $i$-pawn in $m$, and $1\leq j \leq k_{2}$.
\end{enumerate}
Observe that $\text{dom}(\Bmf_{1})\cup \text{dom}(\Bmf_{2})$ is of double exponential size in $\varphi,\psi$. To simplify notation we 
\begin{itemize}
	\item denote copies of types $t$ by $t'$ and copies of pawn mosaics $m$ by $m'$;
	\item often regard a king mosaic $m$ as a copy $m'$ of itself and an $i$-king $t$ in a king mosaic $m$ as a copy $t'$ of itself.
\end{itemize}
We aim to construct $\Bmf_{1}$ and $\Bmf_{2}$ such that the following two conditions hold (where, as announced, $t'$ and $m'$ also range over $i$-kings and king mosaics, respectively):
\begin{enumerate}
\item Any pair $(t',m')$ realizes the 1-type of which $t'$ is a copy. More precisely, for $i=1,2$, if $(t',m')\in \text{dom}(\Bmf_{i})$ and $t'$ is a copy of 1-type $t$ and $\beta(x)\in \text{cl}(\Xi)$, then 
$$
\Bmf_{i} \models \beta(t',m') \quad \Leftrightarrow \quad \beta(x) \in t.
$$
\item For any copy $m'$ of a mosaic, all $(t',m')$ are FO$^{2}(\tau)$-bisimilar. More precisely, for all $t_{1}',t_{2}', m'$ such that $(t_{1}',m')\in \text{dom}(\Bmf_{i})$ and $(t_{2}',m')\in \text{dom}(\Bmf_{j})$
for some $i,j\in \{1,2\}$: $\Bmf_{i},(t_{1}',m') \sim_{\text{FO}^{2},\tau}\Bmf_{j},(t_{2}',m')$.
\end{enumerate}
We first define the interpretation of relation symbols on singleton subsets of $\text{dom}(\Bmf_{i})$ in the obvious way by setting $(t',m')\in R^{\Bmf_{i}}$
iff $R(x)\in t$, for $R$ unary, and $((t',m'),(t',m'))\in R^{\Bmf_{i}}$
iff $R(x,x)\in t$, for $R$ binary.
It thus remains to define the link-types
$l_{\Bmf_{i}}((t_{1}',m_{1}'),(t_{2}',m_{2}'))$ between distinct nodes $(t_{1}',m_{1}')$ and $(t_{2}',m_{2}')$ in $\Bmf_{i}$, $i=1,2$. To this end,
we will carefully associate
\begin{itemize}

  \item with every copy $m'$ of a mosaic a \emph{generator}
    $g\in \text{dom}(\Amf_{1}) \cup \text{dom}(\Amf_{2})$ such
    that $m'$ is a copy of $m=m(g)$;

  \item with every node $(t',m')\in \text{dom}(\Bmf_{i})$ a \emph{witness} $d\in \text{dom}(\Amf_{i})$ for $(t',m')$ such that
    $d\sim_{\text{FO}^{2},\tau} g$ for the generator $g$ of $m'$ and $t'$ is a copy of $t_{\Amf_{i}}(d)$. 
\end{itemize} 
If $(t_{1}',m_{1}')$ and $(t_{2}',m_{2}')$ contain a new $i$-king, then we will define $l_{\Bmf_{i}}((t_{1}',m_{1}'),(t_{2}',m_{2}'))$ as $l_{\Amf_{i}}(d_{1},d_{2})$ for the selected witnesses $d_{1}$ and $d_{2}$ for 
$(t_{1}',m_{1}')$ and $(t_{2}',m_{2}')$, respectively. For $i$-pawns,
$l_{\Bmf_{i}}((t_{1}',m_{1}'),(t_{2}',m_{2}'))$ will be defined using `global' constraints and will not in general be the induced link-type from $\Amf_{i}$.
We now give the detailed construction.

For king mosaics $m$ we simply select as its generator any $g$ with
$m=m(g)$ and for new $i$-kings $(t,m)$ we take the unique
$d\sim_{\text{FO}^{2},\tau} g$ with $t=t_{\Amf_{i}}(d)$ as its
witness. The definition of link-types between new $i$-kings is then as announced: if $(t_{1},m_{1})$ and $(t_{2},m_{2})$ are new distinct $i$-kings, then set $l_{\Bmf_{i}}((t_{1},m_{1}),(t_{2},m_{2})):= l_{\Amf_{i}}(d_{1},d_{2})$ for the witnesses $d_{1},d_{2}$ for $(t_{1},m_{1})$ and $(t_{2},m_{2})$, respectively.

{\bf Link-types between new $i$-kings and semi $i$-pawns}.
Assume $d$ is the witness for an $i$-king $(t,m)$ and $(d,d')$
satisfies $(l,s)$ for a link-type $l$ and 1-type $s$, where $d'$ is a node realizing an $i$-pawn and $m(d')$ is a king-mosaic. Then we aim to ensure that the link-type realized by $((t,m),(s',m(d'))$ equals $l$, for some copy $s'$ of $s$. To obtain these link-types we carefully choose the witnesses $d'$ for semi-$i$-pawns and then take, as announced, the link-type between the selected witnesses as given by $\Amf_{i}$. 

\smallskip
\noindent
(P1) Do the following for all new $i$-kings $(t,m)$: if $(s,n)$ is a
	  pair such that $n$ is a king mosaic and $s$ an $i$-pawn in
	  $n$, $l$ is a link-type, and $d$ is the witness of
	  $(t,m)$ such that $(d,d')$ satisfies $(l,s)$ for some $d'$ with
	  $n=m(d')$, then pick a copy $s'$ of $s$, pick such a $d'$ as
	  the witness for $(s',n)$, and set
	  $l_{\Bmf_{i}}((t_{\Amf_{i}}(e),m(e)),(s',n)):=
	  l_{\Amf_{i}}(e,d')$, for all witnesses $e$ for new
	  $i$-kings.

\smallskip
Note that there are sufficiently many fresh copies $s'$ of
1-types $s$ as $k_{1}\geq m_{1}m_{2}$, where $m_{1}$ is the
number of new $i$-kings and $m_{2}$ is the number of link-types. 
For any pair $(s',n)$ with $n$ a king mosaic and $s'$ a copy of an $i$-pawn $s$ in $n$ not selected according to (P1), pick any $d'$ with $n=m(d')$ such that $s'$ is a copy of $t_{\Amf_{i}}(d')$ as the witness for $(s',n)$ and
let $l_{\Bmf_{i}}((t_{\Amf_{i}}(e),m(e)),(s',n)):= l_{\Amf_{i}}(e,d')$, for all witnesses $e$ for new $i$-kings.

{\bf Link-types between semi $i$-pawns.} For any king mosaics 
$m_{1}, m_{2}$ and $i$-pawns 
$t_{1}\in m_{1}$ and $t_{2}\in m_{2}$ 
define 
\begin{align*}
  L_{i}((t_{1},m_{1}),(t_{2},m_{2}))  ={} &
\{ l_{\Amf_{i}}(d_{1},d_{2}) \mid d_{1}\not=d_{2}, \\
      &  \;\;\; t_{1}= t_{\Amf_{i}}(d_{1}), m_{1}=m(d_{1}),\\
      &  \;\;\; t_{2}= t_{\Amf_{i}}(d_{2}), m_{2}=m(d_{2})\}.	
    \end{align*}
(Note that $(t_{1},m_{1})= (t_{2},m_{2})$ is possible.)
Using the fact that the number of copies of any 1-type exceeds $4\times 2^{2s}$,
it is straightforward to define the link-types $l_{\Bmf_{i}}((t_{1}',m_{1}),(t_{2}',m_{2}))$, where $t_{1}'$ and $t_{2}'$ are copies of $t_{1}$ and $t_{2}$, in such a way that the following holds:

\smallskip
\noindent
(P2) If $t_{1}'$ is a copy of $t_{1}$ and $l$ a link-type, then there exists a copy $t_{2}'$ of $t_{2}$ such that $l=l_{\Bmf_{i}}((t_{1}',m_{1}),(t_{2}',m_{2}))$ iff $l\in L_{i}((t_{1},m_{1}),(t_{2},m_{2}))$.   

\smallskip
{\bf Selecting generators for pawn mosaics.}
To define link-types for pairs of nodes that include full $i$-pawns,
we first fix the generators of copies of pawn mosaics as
follows:

\smallskip
(M) Do the following for all $(t',n)$ which are either new $i$-kings or semi $i$-pawns: if $l$ is a link-type, $s$ a 1-type,  
$d$ is the witness for $(t',n)$, and $(d,d')$ satisfies $(l,s)$ for some $d'$ such that $d'\not\sim_{\text{FO}^{2},\tau}g$ for any $g$ generating a king mosaic, then take such a $d'$ and a copy $m'$ of the pawn mosaic $m$ generated by $d'$ and select as generator of $m'$ any $g'$ with $m(g')=m$ and  $d'\sim_{\text{FO}^{2},\tau} g'$. 

\smallskip
Note that there are sufficiently many copies of pawn mosaics as $k_{2}\geq 2m_{1}m_{2}$, where $m_{1}$ is the number of new $i$-kings and semi $i$-pawns and $m_{2}$ is the number of link-types.
For any copy $m'$ of a pawn mosaic $m$ for which no generator has yet been selected in (M) choose an arbitrary $g$ with $m=m(g)$ as a generator.

{\bf Link-types between new $i$-kings and full $i$-pawns.}
These link-types are now defined similarly to the link-types between new $i$-kings and semi $i$-pawns.

\smallskip
\noindent 
(P3) Do the following for all new $i$-king $(t,m)$: if $(s,n)$ is a pair such that $n'$ is a copy of the pawn mosaic $n$ and $s$ is an $i$-pawn in $n$,
$l$ is a link-type, $d$ is the witness of $(t,m)$, $g$ the generator of $n'$,
and $(d,d')$ satisfies $(l,s)$ for some $d'$ with $d'\sim_{\text{FO}^{2},\tau}g$, then pick a copy $s'$ of $s$, pick such a $d'$ as the witness for $(s',n')$, and set $l_{\Bmf_{i}}((t_{\Amf_{i}}(e),m(e)),(s',n')):= l_{\Amf_{i}}(e,d')$, for all witnesses $e$ for new $i$-kings.  

\smallskip
As in (P1), there are sufficiently many copies for this to work and for any full $i$-pawn $(s',n')$ not yet selected, pick any $d'$ with $d'\sim_{\text{FO}^{2},\tau}g$ for the generator $g$ of $n'$ and $s'$ a copy of $t_{\Amf_{i}}(d')$ as the witness for $(s',n')$ and
let $l_{\Bmf_{i}}((t_{\Amf_{i}}(e),m(e)),(s',n)):= l_{\Amf_{i}}(e,d')$, for all witnesses $e$ for new $i$-kings.

\medskip

{\bf Link-types between semi $i$-pawns and full $i$-pawns.}
For any king mosaic $m$, $i$-pawn $t\in m$, copy $n'$ of
a pawn mosaic $n$, and any $i$-pawn $s\in n$, let $g$ be the generator of 
$n'$ and set
\begin{align*}
  L_{i}((t,m),(s,n')) ={} & 
      \{ l_{\Amf_{i}}(d_{1},d_{2}) \mid t= t_{\Amf_{i}}(d_{1}), m=m(d_{1}), \\
        &\;\;\; s= t_{\Amf_{i}}(d_{2}), d_{2} \sim_{\text{FO}^{2},\tau} g\}
\end{align*}
Similarly to (P2), it is now straightforward to define link-types $l_{\Bmf_{i}}((t',m),(s',n'))$ in such a way that the following holds:

\smallskip
\noindent
(P4) (a) If $t'$ is a copy of $t$ and $l$ a link-type, then there exists a copy $s'$ of $s$ such that  $l=l_{\Bmf_{i}}((t',m),(s',n'))$ iff $l\in L_{i}((t,m),(s,n'))$.
	
(b) If $s'$ is a copy of $s$ and $l$ a link-type, then there exists a copy $t'$ of $t$ such that  $l=l_{\Bmf_{i}}((t',m),(s',n'))$ iff $l\in L_{i}((t,m),(s,n'))$.

\smallskip

{\bf Link-types between full $i$-pawns.} For any pawn mosaics $m_{1},m_{2}$ and $i$-pawns $t_{1}\in m_{1}$ and $t_{2}\in m_{2}$, define $L_{i}((t_{1},m_{1}),(t_{2},m_{2}))$ in exactly the same way as in the definition of link-types between semi $i$-pawns. Then one can define 
the link-types $l_{\Bmf_{i}}((t_{1}',m_{1}'),(t_{2}',m_{2}'))$ in such a way that the following holds:

\smallskip
\noindent
(P5) If $m_{1}'$ and $m_{2}'$ are copies of $m_{1}$ and $m_{2}$, $t_{1}'$ is a copy of $t_{1}$, and $l$ is a link-type, then there exists a copy $t_{2}'$ of $t_{2}$ such that $l=l_{\Bmf_{i}}((t_{1}',m_{1}'),(t_{2}',m_{2}'))$ iff $l\in L_{i}((t_{1},m_{1}),(t_{2},m_{2}))$.   

\smallskip
This finishes the definition of $\Bmf_{1}$ and $\Bmf_{2}$. It is not difficult to show that Conditions~1 and 2 above hold. Assume w.l.o.g. that $\abf_{1}=(a_{11},a_{12})$ and $\abf_{2}=(a_{21},a_{22})$ with $a_{11}\not=a_{12}$. Then, $a_{21}\not=a_{22}$ as otherwise
$\Amf_{1},\abf_{1} \not\sim_{\text{FO}^2,\tau} \Amf_{2},\abf_{2}$.
It is straightforward to ensure in the construction of $\Bmf_{1}$ and $\Bmf_{2}$ above that
$a_{11},a_{12},a_{21},a_{22}$ are witnesses for domain elements 
$(t_{11},m_{1}),(t_{12},m_{2})$ of $\Bmf_{1}$ and $(t_{21},m_{1}),(t_{22},m_{2})$ of $\Bmf_{2}$ and that
$l_{\Bmf_{i}}((t_{i1},m_{1}),(t_{i2},m_{2}))=l_{\Amf_{i}}(a_{i1},a_{i2})$,
	for $i=1,2$. Then we have that 
	$\Bmf_{1}\models \varphi((t_{11},m_{1}),(t_{12},m_{2}))$ and $\Bmf_{2}\models \psi((t_{21},m_{1}),(t_{22},m_{2}))$, by Condition~1. 
By Condition~2, 
$\Bmf_{1},(t_{11},m_{1}),(t_{12},m_{2})\sim_{\text{FO}^2,\tau} \Bmf_{2},(t_{21},m_{1}),(t_{22},m_{2})$.

\subsection{Lower Bound}

The lower bound proof is essentially a modification of the lower bound
for (the bounded arity case for) GF. In fact, it is also a reduction from the word problem of
exponentially space bounded ATMs which uses the same signature $\tau$. 
Again, we aim to construct an FO$^2$-sentence $\varphi'$ such that such an ATM $M$
accepts input $w$ iff $\varphi'\wedge A(x)$ and $\varphi'\wedge \neg A(x)$ are
jointly FO$^2(\tau)$-consistent. The sentence $\varphi'$ is obtained
from the sentence $\varphi$ constructed for GF by replacing the first
conjunct 
$\varphi_0$ with $\varphi_0'$ (recall that all other conjuncts are 
already in $\text{FO}^2$). Recall that 
$\varphi_0$ enforced a cycle of length three using a ternary relation,
which is impossible in FO$^2$. Instead,
we proceed similar to Example~\ref{ex:fo2-beth}. Indeed, $\varphi_0'$
enforces that every element satisfying $A$ is involved
in such a cycle:
\begin{align*}
  \varphi_0' ={} & 
  \forall x\, (Y(x)\rightarrow X(x) \wedge \varphi_3(x)\wedge
    \forall y\,(Y(y)\rightarrow x=y)) \wedge{} \\
    & \forall x\, (A(x)\rightarrow Y(x) )
\end{align*}
where $\varphi_3$ is as in Example~\ref{ex:fo2-beth}, that is, it
enforces the
existence of a path of length three to an element satisfying
$Y$, which is enforced to be a singleton.
%
Now, if $\varphi'\wedge A(x)$ and $\varphi'\wedge \neg A(x)$ are jointly
FO$^2(\tau)$-consistent, there exist models \Amf and \Bmf of
$\varphi'$
and elements $a,b$ such that $a\in A^\Amf$, $b\notin A^\Bmf$, and
$\Amf,a\sim_{\text{FO}^2,\tau} \Bmf,b$. If the latter holds, then from
$a\in A^\Amf$ and $\varphi_0'$ it follows that $b$ has an infinite
outgoing path $\rho$ along $R$ on which every third element satisfies
$X$. As FO$^2(\tau)$-bisimilarity is an
equivalence relation, all these elements satisfying $X$ are actually
FO$^2(\tau)$-bisimilar. Now, the synchronization of the successor
configurations works in the very same way as for GF; we prove
correctness in the appendix.
\begin{restatable}{lemma}{lemcorrectfotwo}\label{lem:correct2exp-fo2}
  $M$ accepts the input $w$ iff there exists models $\Amf,\Bmf$ of
  $\varphi'$
  and elements $a\in A^\Amf$, $b\notin A^\Bmf$ 
  such that $\Amf,a\sim_{\text{FO}^2,\tau}\Bmf,b$.
\end{restatable}
To prove the second part of Theorem~\ref{lem:lower-bounds} (ii) we replace in $\varphi'$ every occurrence of any formula of the form $E(x)$ and $E(y)$ for a unary symbol $E\in \text{sig}(\varphi')\setminus (\tau \cup \{A\})$ by the formula 
$$
\chi_{E}(x) = \exists y(R_{E}(x,y) \wedge \exists x (N(y,x) \wedge \exists y (N(x,y)\wedge A(y))))
$$
and the formula $\chi_{E}(y)$ obtained from $\chi_{E}(x)$ by swapping $x$ and $y$, respectively. Here $R_{E}$, $E\in \text{sig}(\varphi')\setminus (\tau \cup \{A\})$, and $N$ are fresh binary relation symbols. An analogue of Lemma~\ref{lem:correct2exp-fo2} is proved in the appendix for the resulting formula $\varphi''$ and the signature $\tau'$ containing all relation symbols in $\varphi''$ except $A$.

%
%
%

\section{Conclusion}

We have shown tight complexity bounds for interpolant and explicit
definition existence in GF and \coNTwoExpTime/\TwoExpTime upper and,
respectively, lower bounds for FO$^2$. Many questions remain to be
explored. First we conjecture that these problems are
actually \coNTwoExpTime-complete in FO$^{2}$. Then it would be of interest to
determine the size of interpolants/explicit definitions in GF and
FO$^{2}$ if they exist. Note that recently the size and computation of
interpolants in GNF has been studied in
depth~\cite{DBLP:journals/tocl/BenediktCB16}. In contrast to GF, GNF
enjoys CIP and PBDP and it is not difficult to show using the
complexity lower bound proof given above that in GF minimal
interpolants/explicit definitions are, in the worst case, at least by
one exponential larger than in GNF. 

There are many logics without the CIP and PBDP for which the complexity of interpolant and explicit definition existence remain to be explored, examples include the extension of FO$^{2}$ with counting, FO$^{2}$ without equality, the extension of GF with constants, and the Horn fragment of GF introduced in~\cite{DBLP:conf/lics/JungPWZ19}.

\section*{Acknowledgment}
Frank Wolter was supported by EPSRC grant EP/S032207/1.



\cleardoublepage

\cleardoublepage
\appendix
	
\section*{Proofs for Section~\ref{sec:one}}

\medskip
\noindent
\lemcharinterpolation*

\noindent
\begin{proof} 
	$(\Leftarrow)$ Assume there is an \Lmc-interpolant
	$\theta(\xbf)$ and let $\Amf,\Bmf$ be structures and $\abf,\bbf$
	be tuples such that $\Amf\models\varphi(\abf)$ and
	$\Bmf\models\neg\psi(\bbf)$.  Suppose further that $\Amf,\abf\sim_{\Lmc,\tau}\Bmf,\bbf$. Since $\varphi(\xbf)\models \theta(\xbf)$, we have
	$\Amf\models\theta(\abf)$. By Lemma~\ref{lem:guardedbisim}, we
	obtain $\Bmf\models \theta(\bbf)$. Finally, as
	$\theta(\xbf)\models \psi(\xbf)$, we obtain
	$\Bmf\models\psi(\bbf)$, a contradiction.

	$(\Rightarrow)$ Suppose that for all structures $\Amf,\Bmf$ and
	tuples $\abf,\bbf$ such that $\Amf\models\varphi(\abf)$ and
	$\Bmf\models\neg\psi(\bbf)$ we have $\Amf,\abf \not\sim_{\Lmc,\tau}\Bmf,\bbf$. Let $\Phi$ be defined by
	taking
	\[\Phi= \{\varphi'(\xbf)\in \Lmc(\tau)\mid
	\varphi(\xbf)\models\varphi'(\xbf)\}.\]
	Clearly, $\varphi(\xbf)\models\Phi$.  We claim that also
	$\Phi\models\psi(\xbf)$. To see this, let $\Bmf,\bbf$ such that
	$\Bmf\models\Phi(\bbf)$. Let $\Bmf'$ be an
	$\omega$-saturated elementary extension of $\Bmf$ and let $\Amf,\abf$ be an 
	$\omega$-saturated pointed structure realizing $\{
	\chi(\xbf)\in \Lmc(\tau)\mid \Bmf\models \chi(\bbf)\}\cup \{\varphi\}$ in $\abf$
	($\Amf,\abf$ exists by compactness and the definition of $\Phi$). By definition of $\Phi$ and
	Lemma~\ref{lem:guardedbisim}, we have
	$\Amf,\abf\sim_{\Lmc,\tau}\Bmf',\bbf$. By the initial assumption, we
	cannot have $\Bmf'\models\neg\psi(\bbf)$ and thus
	$\Bmf\models\psi(\bbf)$. By compactness, there is a finite subset
	$\Phi'$ of $\Phi$ such that $\Phi'\models\psi(\xbf)$. The
	conjunction of the formulas in $\Phi'$ is the required interpolant.
\end{proof}	

\lemdeftoint*

\noindent \begin{proof} 
  Assume $\vp$, $\theta(\xbf)$, and $\tau$ are given. Then
  $\theta(\xbf)$ is explicitly definable under $\vp$ iff there exists
  an $\Lmc$-interpolant for $\vp \wedge \theta(\xbf),\vp' \rightarrow
  \theta'(\xbf)$, where $\vp'$ and $\theta'$ are obtained from
  $\varphi$ and $\theta$, respectively, by renaming all non-$\tau$
  symbols $R$ to fresh $R'$ of the same arity. 	
\end{proof}

\section*{Proofs for Section~\ref{sec:gf}}

\medskip
\noindent
{\bf Claim 1.}
{\it For all $i,j$ with $1\leq i, j\leq m$,
we have: \begin{enumerate}
	\item For every $\sigma\in T_{i}$ with
	$\text{tail}(\sigma)=(t(\ybf),\Phi)$, we have $\Amf_{i}\models
	t(v_{\sigma}(\ybf))$;
	
	\item Let $H_{i,j}$ be the set of all mappings
	$p_{\sigma,\sigma',\zbf}$, where 
	\begin{itemize}
		
		\item $\sigma\in T_i,\sigma'\in T_j$, 
		$\text{tail}(\sigma)=(t(\ybf),\Phi)$, and
		$\text{tail}(\sigma')=(t'(\ybf'),\Phi)$; 
		
		\item $\zbf$ is a tuple with $[\zbf] \subseteq [\ybf]\cap
		[\ybf']$ and $v_{\sigma}(\zbf)$ is $\tau$-guarded in $\Amf_i$
		$($or, equivalently, $v_{\sigma'}(\zbf)$ is $\tau$-guarded in $\Amf_j$$)$;
		
		\item $p_{\sigma,\sigma',\zbf}: v_{\sigma}(\zbf) \mapsto
		v_{\sigma'}(\zbf)$.  
		
	\end{itemize} 
	Then $H_{i,j}$ is a GF($\tau$)-bisimulation between $\Amf_{i}$
	and $\Amf_{j}$. 
\end{enumerate}
}

\medskip
\noindent
\begin{proof}
For Point~1, we prove by induction that, for all $\sigma\in T_{i}$ with
$\text{tail}(\sigma)=(t(\ybf),\Phi)$ and all formulas
$\varphi(\zbf)$ with $[\zbf]\subseteq [\ybf]$, we have: 
\[\varphi(\zbf)\in t(\ybf)\quad \text{iff}\quad \Amf_{i}\models
\varphi(v_\sigma(\zbf))\]
The induction base is given by the definition of
$\text{bag}_{i}(\sigma)$. If $\varphi$ is of
the shape $\neg\varphi'$, $\varphi'\wedge\varphi''$, or
$\varphi'\vee\varphi''$, the statement is immediate from the
hypothesis. Consider now $\varphi(\zbf)=\exists \xbf
(R(\zbf,\xbf)\wedge \lambda(\zbf,\xbf))$. 

\smallskip $(\Rightarrow)$ Since \Mmc is existentially saturated,
there is a $\Phi'\in \Mmc$ such that $\Phi,\Phi'$ are compatible and
$R(\zbf,\xbf')\wedge\lambda(\zbf,\xbf')\in t'(\ybf')$ for some
$t'(\ybf')\in \Phi'$ such that $t(\ybf)$ and $t'(\ybf')$ coincide on $[\ybf]\cap [\ybf']$. By definition of $T_{i}$ and compatibility of
$\Phi,\Phi'$, we have $\sigma'=\sigma\cdot(t'(\ybf'),\Phi')\in T_{i}$. Moreover,
by induction, we obtain that $\Amf_{i}$ satisfies
$R(\zbf,\xbf')\wedge\lambda(\zbf,\xbf')$ under $v_{\sigma'}$. By
definition of $\text{bag}_{i}(\sigma)$ and $\text{bag}_{i}(\sigma')$, we get
$\Amf_{i}\models\varphi(v_\sigma(\zbf))$.

\smallskip $(\Leftarrow)$ Conversely, assume
$\Amf_{i}\models\varphi(v_\sigma(\zbf))$. By construction, there is some
$\sigma'\in T_{i}$ such that $v_{\sigma}(\zbf)=v_{\sigma'}(\zbf)$ and
$\Amf_{i}$ satisfies $R(\zbf,\xbf')\wedge \lambda(\zbf,\xbf')$ under
$v_{\sigma'}$, for some $\xbf'$. By induction hypothesis,
$R(\zbf,\xbf')\wedge \lambda(\zbf,\xbf')\in t'(\ybf')$, where
$\text{tail}(\sigma')=(t'(\ybf'),\Phi')$.
Thus, $\exists
\xbf (R(\zbf,\xbf)\wedge \lambda(\zbf,\xbf))= \varphi(\zbf)\in t'(\ybf')$. As
$v_{\sigma}(\zbf)=v_{\sigma'}(\zbf)$, the construction of $T_{i}$
implies that $t'(\ybf')$ and $t(\ybf)$ coincide on all subformulas
over $\zbf$, hence $\varphi(\zbf)\in t(\ybf)$.

\medskip

For Point~2, observe first that the $p_{\sigma,\sigma',\zbf}$ are
partial $\tau$-isomorphisms between $\tau$-guarded tuples 
since all $\Phi\in \mathcal{M}$ are $\tau$-uniform. (In addition, the observation that $v_{\sigma}(\zbf)$ is $\tau$-guarded in $\Amf_{i}$ iff $v_{\sigma'}(\zbf)$ is $\tau$-guarded in $\Amf_{j}$ follows from the condition that $\Phi$ is $\tau$-uniform.) 
By symmetry,
it suffices to prove Condition~(i) for GF($\tau$)-bisimulations. 

Let $p\in H_{i,j}$. Then we have $\sigma\in T_i,\sigma'\in T_j$ with 
$\text{tail}(\sigma)=(t(\ybf),\Phi)$ and
$\text{tail}(\sigma')=(t'(\ybf'),\Phi)$ and we have a tuple $\zbf$ such that 
$[\zbf] \subseteq [\ybf]\cap [\ybf']$ and $v_{\sigma}(\zbf)$
is $\tau$-guarded in $\Amf_i$ and $p=p_{\sigma,\sigma',\zbf}$.
Consider any tuple $\bbf$ with $\Amf_i\models R(\bbf)$ for some $R\in \tau$.
We have to show that there exists a mapping $p_{\rho,\rho',\zbf'}\in
H_{i,j}$ with domain $[\bbf]$ which coincides with $p_{\sigma,\sigma',\zbf}$ on $[v_{\sigma}(\zbf)] \cap [\bbf]$. We distinguish on whether or not
that intersection is empty. 

\smallskip
\emph{Case 1.} $[v_{\sigma}(\zbf)] \cap [\bbf]=\emptyset$. The existence of such a mapping
follows from GF($\tau$)-bisimulation saturatedness: to see this, observe
that, as we have a tree decomposition, there exists $\rho_{0}\in
T_{i}$ such that $[\bbf] \subseteq \text{dom}(\text{bag}(\rho_{0}))$.
Let $\text{tail}(\rho_{0})=(s(\xbf_{0}),\Omega)$. Then there exists a
tuple $\ybf_{0}$ with $[\ybf_{0}] \subseteq [\xbf_{0}]$ such that
$v_{\rho_0}(\ybf_{0})= \bbf$. We have $R(\ybf_{0})\in s(\xbf_{0})$. As
$\hat{t}_{j}\in \Omega$, by GF($\tau$)-bisimulation saturatedness of
$\Omega$, there exists $s'(\ybf_{0}')\in \Omega$ such that
$\hat{t}_{j}\subseteq s'(\ybf_{0}')$ and $[\ybf_{0}']=[\ybf_{0}]$. But then
$R(\ybf_{0})\in s'(\ybf_{0}')$. Also $\rho= (\hat{t}_{j},\widehat{\Psi})\cdot
(s'(\ybf_{0}'),\Omega) \in T_{j}$.  Thus $p_{\rho_{0},\rho,\ybf_{0}'}$ is as
required.

\smallskip\textit{Case 2.} $[v_{\sigma}(\zbf)] \cap [\bbf]\not=\emptyset$. As we have a tree
decomposition, there exists $\rho_{0}\in T_{i}$ such that $[\bbf]
\subseteq \text{dom}(\text{bag}(\rho_{0}))$. Let
$\text{tail}(\rho_{0})=(s(\xbf_{0}),\Omega)$. Then there exists a tuple
$\zbf'$ with $[\zbf'] \subseteq [\xbf_{0}]$ such that
$v_{\rho_{0}}(\zbf')= \bbf$.  
We distinguish the following cases:
\begin{enumerate}[label=(\alph*)]
	
	\item $\rho_{0}=\sigma$;
	
	\item $\rho_0\neq \sigma$.
	%

\end{enumerate}
Assume first that (a) holds. Then $(s(\xbf_{0}),\Omega) = (t(\ybf),\Phi)$ and $\bbf = v_{\sigma}(\zbf')$. We use
GF($\tau$)-bisimulation saturatedness of $\Phi$. Consider the restriction
$\zbf''$ of $\zbf'$ to $[\zbf]\cap [\zbf']$ and the restriction
$t'(\ybf')_{|[\zbf'']}$ of $t'(\ybf')$ to $[\zbf'']$. Then there
exists $s'(\zbf_{0}')\in \Phi$ such that $t'(\ybf')_{|[\zbf'']} \subseteq
s'(\zbf_{0}')\in \Phi$ and $[\zbf_{0}']=[\zbf']$. Let $\sigma''= \sigma'\cdot (s'(\zbf_{0}'),\Phi)\in
T_j$.  Then $p_{\sigma,\sigma'',\zbf_{0}'}$ is as required, as $\Phi$ is 
$\tau$-uniform.

\medskip

Assume now that Point~(b) holds.
Consider the restriction $\zbf''$ of $\zbf'$ to $[\zbf]\cap [\zbf']$
and the restriction $t'(\ybf')_{|[\zbf'']}$ of $t'(\ybf')$ to
$[\zbf'']$. Consider the restriction $\Phi_{|[\zbf'']}$ of $\Phi$ to
$[\zbf'']$. By closure under restrictions, $\Phi_{|[\zbf'']}\in
\mathcal{M}$. Observe that $\Phi,\Phi_{|[\zbf'']}$ and
$\Phi_{|[\zbf'']}, \Omega$ are compatible:
indeed, in the tree
decomposition all bags on the path from $\sigma$ to $\rho_0$ 
have a
tail $(\cdot,\Omega')$ satisfying $\Phi_{|[\zbf'']}\subseteq \Omega'$. Thus
$t'(\ybf')_{|[\zbf'']}\in \Omega$.
Using the fact that $\Omega$ is GF($\tau$)-bisimulation saturated, 
one can now show that there exists $s'(\zbf_{0}')\in \Omega$ such that $t'(\ybf')_{|[\zbf'']} \subseteq  s'(\zbf_{0}')$ and
$[\zbf_{0}']=[\zbf']$. We then have 
\[
\rho =  \sigma' \cdot (t'(\ybf')_{|[\zbf'']},\Phi_{|[\zbf'']}) \cdot
(s'(\zbf_{0}'),\Omega)\in T_j
\]
and $p_{\rho_{0},\rho,\zbf_{0}'}$ is as required.
\end{proof}

\medskip
\noindent
{\bf Lemma~\ref{lem:readoff}}
{\it Let $\Amf_{1},\abf_{1}$ and $\Amf_{2},\abf_{2}$ be pointed
	structures with $\abf_{1}$ and $\abf_{2}$ tuples with pairwise
	distinct elements of length $m\leq \text{fv}(\Xi)$ and let $\tau$ be
	a signature. Consider assignments $\xbf_{0} \mapsto \abf_{i}$ with
	$[\xbf_{0}]\subseteq \{x_{0},\ldots,x_{2n}\}$. If $\Amf_{1},\abf_{1}
	\sim_{\text{GF},\tau} \Amf_{2},\abf_{2}$, then there exists an existentially saturated set $\mathcal{M}$ of $\tau$-mosaics and some $\Psi\in \Mmc$ such that
	\begin{itemize} 
		
		\item all $\Phi\in \mathcal{M}$ with $\Phi\neq \Psi$ use at most
		$\text{width}(\Xi)$ many free variables; 
		
		
		\item there exist types $t_1(\xbf_0),t_2(\xbf_0)\in \Psi$ such that
		$t_i(\xbf_0)=\text{tp}(\Amf_{i},\xbf_{0} \mapsto \abf_{i})$ for
		$i=1,2$ and
		all types $t(\ybf)\in \Psi\setminus\{t_1(\xbf_0),t_2(\xbf_0)\}$
		use at most $\text{width}(\Xi)$ free variables among
		$[\xbf_0]$.
		
	\end{itemize}
}

\medskip
\noindent
\begin{proof} Assume w.l.o.g. that $\Amf_{1}$ and
	$\Amf_{2}$ are disjoint. For any tuples $\bbf_{1}$ in $\Amf_{i}$ and
	$\bbf_{2}$ in $\Amf_{j}$ with $i,j\in \{1,2\}$, we use
	$\text{tp}(\xbf_{1}\mapsto \bbf_{1})$ to denote
	$\text{tp}(\Amf_{i},\xbf_{1} \mapsto \bbf_{1})$ and we write
	$\bbf_{1}\sim_{\text{GF},\tau} \bbf_{2}$ if
	$\Amf_{i},\bbf_{1}\sim_{\text{GF},\tau} \Amf_{j},\bbf_{2}$.  Define
	$\mathcal{M}$ as follows. Take any tuple $\abf$ of distinct elements
	in $\Amf_{i}$, $i\in \{1,2\}$. Take a tuple $\xbf$ from
	$\{x_{1},\ldots,x_{2n}\}$ such that $v: \xbf \mapsto \abf$ is a
	bijection. Then let $\Phi_{\abf,\xbf}$ contain all types
	$\text{tp}(v':\xbf_{|Y} \mapsto \bbf)$ with $Y\subseteq [\xbf]$
	and $\bbf$ in either $\Amf_{1}$ or $\Amf_{2}$ such that
	$v(\xbf_{|Y}) \sim_{\text{GF},\tau} v'(\xbf_{|Y})$.
	
	Let $\mathcal{M}$ contain all such $\Phi_{\abf,\xbf}$ with $\abf$ of length at most $\text{width}(\Xi)$ and $\xbf$ from
	$\{x_{1},\ldots,x_{2n}\}$. Moveover, if $m>\text{width}(\Xi)$, then
	add $\widehat{\Phi}_{\abf_{1},\xbf_{0}}$ to $\mathcal{M}$,
	where $\widehat{\Phi}_{\abf_{1},\xbf_{0}}$ is obtained from 
	$\Phi_{\abf_{1},\xbf_{0}}$ by removing all $t$ distinct from 
	$t_{1}(\xbf_{0})$ and $t_{2}(\xbf_{0})$ using more than $\text{width}(\Xi)$ many free variables. 
	
	We show that $\mathcal{M}$ is as required. By definition,
	$\text{tp}(\Amf_{1},\xbf_{0} \mapsto \abf_{1})$, $\text{tp}(\Amf_{2},\xbf_{0}\mapsto \abf_{2})\in \Phi_{\abf_{1},\xbf_{0}}
	\in \mathcal{M}$. 
	
	For the next steps we first assume that instead of  $\widehat{\Phi}_{\abf_{1},\xbf_{0}}$ we have $\Phi_{\abf_{1},\xbf_{0}}$ in $\mathcal{M}$. Then observe that if we have any $\Phi\in \mathcal{M}$ and $t(\xbf'),s(\xbf'')\in \Phi$, then we can assume that $\Phi=\Phi_{\abf,\xbf}$, we have a bijection $v$ from $\abf$ to $\xbf$, $\xbf'= \xbf_{|Y'}$ and $\xbf''= \xbf_{|Y''}$ for appropriate
	sets of variables $Y',Y''\subseteq [\xbf]$, and there are $v':\xbf_{|Y'} \mapsto \Amf_{i}$ and $v'':\xbf_{Y''}\mapsto \Amf_{j}$ such that 
	$v'(\xbf_{|Y'}) \sim_{\text{GF},\tau} v(\xbf_{|Y'})$
	and $v''(\xbf_{|Y''}) \sim_{\text{GF},\tau} v(\xbf_{|Y''})$.
	Then $v'(\xbf_{|Y'\cap Y''}) \sim_{\text{GF},\tau} v''(\xbf_{|Y'\cap Y''})$. We show that each $\Phi_{\abf,\xbf}$ is $\tau$-uniform and GF($\tau$)-bisimulation saturated.     
	\begin{enumerate}
		\item Every $\Phi_{\abf,\xbf}\in \mathcal{M}$ is $\tau$-uniform: let $t(\xbf'),s(\xbf'')\in \Phi_{\abf,\xbf}$ be as above and assume that $Q(\zbf)$ is a $\tau$-guard with
		$[\zbf]\subseteq [\xbf']\cap [\xbf'']$. Then $[\zbf]\subseteq Y'\cap Y''$ and so $Q(\vec z)\in t(\xbf')$ iff $Q(\zbf)\in s(\xbf'')$ since
		$v'(\xbf_{|Y'\cap Y''}) \sim_{\text{GF},\tau} v''(\xbf_{|Y'\cap Y''})$, as required.  
		\item To show GF($\tau$)-bisimulation saturatedness let $\Phi_{\abf,\xbf}\in \mathcal{M}$ and $t(\xbf'),s(\xbf'')\in \Phi_{\abf,\xbf}$ be as above
		and let $R(\ybf)\in t(\xbf')$
		with $[\xbf'']\subseteq [\ybf]$ be a strict $\tau$-guard. We have $Y''\subseteq [\ybf] \subseteq Y'$ and $v'(\xbf_{|Y''}) \sim_{\text{GF},\tau} v''(\xbf_{|Y''})$. Let $H$ be the  GF($\tau$)-bisimulation witnessing this. By the definition of  GF($\tau$)-bisimulations, there exists
		$p\in H$ with domain $v'(\xbf_{|[\ybf]})$ such that $p\circ v'_{|Y''} = v''$. Now we expand $v''$ to the domain $[\ybf]$ by setting $\hat{v}:= p\circ v'_{|{\xbf_{|[\ybf]}}}$.
		Let $\bbf'$ be the image of $\xbf_{|[\ybf]}$ under $\hat{v}$.
		Then the type $\text{tp}(\hat{v}:\xbf_{|[\ybf]}\mapsto \bbf')$ is as required.  		
	\end{enumerate} 
	Finally we show that every $\Phi\in \mathcal{M}$ is existentially saturated
	in $\mathcal{M}$. 
	Assume $\Phi_{\abf,\xbf}$ is given. 
	Assume $\exists \ybf (R(\xbf',\ybf)\wedge
	\lambda(\xbf',\ybf)) \in 
	t(\xbf_{|Y}) = \text{tp}(v':\xbf_{|Y} \mapsto \bbf)$
	with $Y\subseteq [\xbf]$ and $\bbf$ w.l.o.g. in $\Amf_{1}$.
	Then $\Amf_{1}\models_{v'} \exists \ybf (R(\xbf',\ybf)\wedge \lambda(\xbf',\ybf))$.
	Then we find an assignment $v''$ for the variables in $[\xbf'\ybf]$ 
	which coincides with $v'$ on $[\xbf']$ such that  $\Amf_{1}\models_{v''} R(\xbf',\ybf)\wedge \lambda(\xbf',\ybf)$. 
	Take a tuple $\cbf$ of distinct elements with $[\cbf]=[v''(\xbf'\ybf)]$
	and a tuple $\ybf'$ of variables in $\{x_{1},\ldots,x_{2n}\}$
	such that $[\xbf'] = [\xbf] \cap [\ybf']$ and we have a bijection $\rho:\ybf'\mapsto \cbf$ 
	which coincides with $v'$ on $[\xbf']$. Then $\rho(\ybf'_{|[\xbf']}) \sim_{\text{GF},\tau} v(\xbf_{|[\xbf']})$ and so $\Phi_{\abf,\xbf}$ and $\Phi_{\cbf,\ybf'}$
	are compatible and $\Phi_{\cbf,\ybf'}$ is as required.
	
	For the proof with $\widehat{\Phi}_{\abf_{1},\xbf_{0}}$ instead of $\Phi_{\abf_{1},\xbf_{0}}$ in $\mathcal{M}$ observe that $\widehat{\Phi}_{\abf_{1},\xbf_{0}}$ is $\tau$-uniform and GF($\tau$)-bisimulation saturated as $\widehat{\Phi}_{\abf_{1},\xbf_{0}}$ behaves in exactly the same way as $\Phi_{\abf_{1},\xbf_{0}}$ regarding $\tau$-guarded $Q(\ybf)$. For the same reason all elements of $\mathcal{M}$ are still existentially saturated in $\mathcal{M}$. 
\end{proof}

\subsection*{\TwoExpTime Lower Bound}

\lemcorrectgf*

\noindent
\begin{proof}
  $(\Rightarrow)$ If $M$ accepts $w$, there is a computation tree of $M$ on $w$. We construct a
single model $\Amf$ of $\varphi$ as follows. Let $\Amf^*$ be the infinite tree-shaped structure
that represents the computation tree of $M$ on $w$ as described above,
that is, configurations are represented by sequences of $2^n$ elements
linked by $S$. Moreover, all elements of a configuration are labeled
with $B_\forall$, $B_\exists^1$, or $B_\exists^2$ depending on whether
the configuration is universal or existential, and in the latter case
the superscript indicates which choice has been made for the
existential state. Finally, the first element of the first successor
configuration of a universal configuration is labeled with $Z$.
In particular, 
$\Amf^*$ only interprets the symbols in $\tau$ non-empty. Now, we
obtain structures $\Amf_k$, $k<2^n$ from $\Amf^*$ by interpreting
non-$\tau$-symbols as follows: 
\begin{itemize}
	
  \item the entire domain of $\Amf_{k}$ satisfies $I$;


  \item the $U$-counter starts at $0$ at the root and counts modulo
    $2^n$ along each $S$-path;

  \item the $V$-counter starts at $k$ at the root and
    counts modulo $2^n$ along each $S$-path;

  \item the auxiliary concept names of the shape $A_\sigma^i$ and
    $A_\sigma'$ are interpreted in a minimal way so as to satisfy the
    sentences listed above. Note that the sentences are Horn, thus
    there is no choice.

\end{itemize}
Now obtain $\Amf$ from $\Amf^*$ and the $\Amf_k$ as follows.  First,
create a both side infinite $R$-path \[\ldots
b_{-2}Rb_{-1}Rb_0Rb_1Rb_2\ldots \] and realize the corresponding
$A$-counter along the path and label every $b_{3k}$,
$k\in\mathbb{Z}$, with $X$. Then, add all
$\Amf_k^*$ to every node $b_{3k}$, $k\in \mathbb{Z}$, on the path by
identifying the roots of the $\Amf_k$ with the respective node on the
path.  Moreover, add to $\Amf$ three elements $a_0,a_1,a_2$ such that
$(a_0,a_1,a_2)\in G^\Amf$, $(a_0,a_1),(a_1,a_2),(a_2,a_0)\in R^\Amf$,
$a_0\in X^\Amf$, and $a_0\in A^\Amf$. Finally, add a copy of $\Amf^*$ to \Amf by identifying the
root of $\Amf^*$ with $a_0$. We claim that $\Amf$ is as required. In
particular, $\Amf,a_0$ is a model of $\varphi\wedge A(x)$,
$\Amf,b_0$ is a model of $\varphi\wedge \neg A(x)$, and the set $S$ of all
mappings
\begin{itemize}

  \item $(a_i,a_{i+1})\mapsto (b_{i+3k},b_{i+3k+1})$ with
    $k\in\mathbb{Z}$, $i\in\{0,1,2\}$, and $a_3:=a_0$,

%
%
%

  \item $(e,f)\mapsto (e',f')$ with $(e,f)\in S^\Bmf$ and $e',f'$
    copies of $e,f$ in some $\Amf_k$, and

  \item all restrictions of the above,

\end{itemize}
is a GF($\tau$)-bisimulation on $\Amf$ with $a_0\mapsto b_0\in S$.

$(\Leftarrow)$ Let $\Amf,\Bmf$ be a models of $\varphi$ such that
$\Amf,a\sim_{\text{GF},\tau}\Bmf,b$ for some elements $a,b$ with
$a\in A^\Amf$, $b\notin A^{\Bmf}$. As it was argued above, due to the
three-element $R$-loop enforced at $a$ via $\varphi_0$,
from $b$ there has to be an
outgoing infinite $R$-path on which all $S$-trees are guarded
$\tau$-bisimilar. (There is also an incoming infinite $R$-path with
this property, but it is not relevant for the proof.) All those
$S$-trees are additionally labeled with some auxiliary relation
symbols not in $\tau$, depending on the distance from $b$. However, it
can be shown using the arguments that accompanied the construction of
$\varphi$ that all $S$-trees contain a computation tree of $M$ on
input $w$. Hence, $M$ accepts $w$.
\end{proof}

\subsection*{\ThreeExpTime Lower Bound}

We show how to axiomatize the predicate $E$ as announced in the main
part, that is, for pairs $a,a'$ and
$b,b'$, where $b,b'$ represents a successor node of $a,a'$,
and for all
$\abf,\abf'\in\{a,a'\}^n$ and $\bbf,\bbf'\in \{b,b'\}^n$, we have
\begin{equation}
  E(\abf\abf'aa'\bbf\bbf'bb')\text{\quad iff\quad }
  r(\abf)=r(\bbf)\text{ and }r(\abf')=r(\bbf').\label{eq:axiomE}
\end{equation}
We abbreviate the tuples $\xbf\xbf'$ and $\ybf\ybf'$ with $\ubf$ and
$\vbf$, respectively; thus $\ubf = u_0\ldots u_{2n-1}$ and $\vbf =
v_0\ldots v_{2n-1}$ are tuples of length $2n$.  Moreover, let $\Sigma$
be the set of all substitutions $[u_i/x,v_i/y]$ and $[u_i/x',v_i/y']$,
for all $i<2n$. Now, add the following sentences: 
\begin{align*}
  & \forall xx'yy'\big( R(xx'yy')\rightarrow
  E(x^{2n}xx'y^{2n}yy')\big) \\
  & \forall \ubf xx'\vbf yy' \big( E(\ubf xx'\vbf
  yy')\rightarrow\bigwedge_{\sigma\in\Sigma}
  E(\sigma(\ubf)xx'\sigma(\vbf)yy')\big) \\
  & \forall \ubf xx'\vbf yy' \big( E(\ubf xx'\vbf
  yy')\rightarrow{} \\
  &\hspace{1cm} \bigwedge_{i<2n} (u_i=x\wedge
  v_i=y)\vee(u_i=x'\wedge v_i=y') \big)
\end{align*}
These sentences axiomatize $E$ as required, since the last sentence
enforces ``only if'' of Property~\eqref{eq:axiomE} while the first and
second sentence together enforce ``if''. 


We finish noting that $\text{Min}_U(xx')$ can be expressed by the
formula
\[\text{Min}_U(xx') = \forall \xbf \big( D_U(\xbf) \rightarrow
\bigvee_{i<n} (x_i\neq x\wedge x_i\neq x')\big).\]

\section*{Proofs for Section~\ref{sec:fo2}}

\subsection*{Proofs for the Upper Bound}

\begin{lemma}
  Joint FO$^{2}(\tau)$-consistency can be reduced in polynomial time
  to joint FO$^{2}(\tau)$-consistency for formulas using relation
  symbols of arity at most two.  
\end{lemma}

\begin{proof}
	We show that the construction given in the finite model property proof in 
	\cite{DBLP:journals/bsl/GradelKV97} also works for joint FO$^{2}(\tau)$-consistency.
	
Consider FO$^{2}(\tau)$-formulas $\varphi$ and $\psi$. We may assume that $\text{sig}(\varphi)\cap \text{sig}(\psi)=\tau$. For any relation symbol $R$ of arity at least three that occurs in $\varphi$ or $\psi$ we do the following: for any atomic formula $R(v_{1},\ldots,v_{n})$ that occurs in $\varphi$ or $\psi$ introduce a fresh relation symbol $R^{v_{1},\ldots,v_{n}}$ of arity two if both $x$ and $y$ occur in $v_{1},\ldots,v_{n}$ and of arity one otherwise. 

If both $x$ and $y$ occur in $v_{1},\ldots,v_{n}$, then replace in $\varphi$ and $\psi$ every occurrence of $R(v_{1},\ldots,v_{n})$ in $\varphi,\psi$ by $R^{v_{1},\ldots,v_{n}}(x,y)$.
If only $x$ occurs in $v_{1},\ldots,v_{n}$, then replace $R(v_{1},\ldots,v_{n})$ by $R^{v_{1},\ldots,v_{n}}(x)$ and if only $y$ occurs in $v_{1},\ldots,v_{n}$ then replace $R(v_{1},\ldots,v_{n})$ by $R^{v_{1},\ldots,v_{n}}(y)$. Let $\varphi'$ and $\psi'$ be the resulting formulas.

It remains to capture the logical relationships between different formulas $R(v_{1},\ldots,v_{n})$ and $R(v_{1}',\ldots,v_{n}')$ using implications
between the fresh atomic formulas. For example, if $R(v_{1},\ldots,v_{n})$ and $R(v_{1}',\ldots,v_{n}')$ are both subformulas of $\varphi$ or $\psi$ and  
$R(v_{1}',\ldots,v_{n}')$ is obtained from $R(v_{1},\ldots,v_{n})$ by replacing $x$ by $y$ and $y$ by $x$, then we take the implication
$$
\forall x \forall y (R^{v_{1},\ldots,v_{n}}(x,y) \leftrightarrow R^{v_{1}',\ldots,v_{n}'}(y,x))
$$
We also take for any $R$ in $\varphi$ or $\psi$ and any two distinct  $R(v_{1},\ldots,v_{n})$ and $R(v_{1}',\ldots,v_{n}')$ occurring in $\varphi$ or $\psi$ the implication:
$$
\forall x (R^{v_{1},\ldots,v_{n}}(x,x) \leftrightarrow R^{v_{1}',\ldots,v_{n}'}(x,x))
$$
Let $\chi_{R}$ be the conjunction of all these implications between the fresh atomic formulas. Now let 
$$
\varphi^{\dagger} = \varphi' \wedge \bigwedge_{R \text{ occurs in } \varphi}\chi_{R},
\quad
\psi^{\dagger} = \psi' \wedge \bigwedge_{R \text{ occurs in } \psi}\chi_{R}
$$
and let $\tau'$ contain all relation symbols of arity at most two in $\tau$ and all fresh $R^{v_{1},\ldots,v_{n}}$ for $R\in \tau$.

We show that $\varphi$ and $\psi$ are jointly FO$^{2}(\tau)$-consistent iff  $\varphi^{\dagger}$ and $\psi^{\dagger}$ are jointly FO$^{2}(\tau')$-consistent. 

Assume $\Amf,\abf \sim_{\text{FO}^{2},\tau} \Bmf,\bbf$, $\Amf\models \varphi(\abf)$, and $\Bmf\models \psi(\bbf)$. Define the structure $\Amf'$ in the same way as $\Amf$ except that for relation symbols $R$ of arity $\geq 3$:
\begin{itemize}
	\item $(a,b)\in (R^{v_{1},\ldots,v_{n}})^{\Amf'}$ if $\Amf\models_{v} R(v_{1},\ldots,v_{n})$ for $v(x)=a$ and $v(y)=b$, if $x$ and $y$ occur in  $v_{1},\ldots,v_{n}$ and $R^{v_{1},\ldots,v_{n}}$ occurs in $\varphi$ or $\psi$.
    \item $a\in (R^{v_{1},\ldots,v_{n}})^{\Amf'}$ if $\Amf\models_{v} R(v_{1},\ldots,v_{n})$ for $v(x)=a$, if only $x$ occurs in $v_{1},\ldots,v_{n}$ and $R^{v_{1},\ldots,v_{n}}$ occurs in $\varphi$ or $\psi$. 
    \item $a\in (R^{v_{1},\ldots,v_{n}})^{\Amf'}$ if $\Amf\models_{v} R(v_{1},\ldots,v_{n})$ for $v(y)=a$, if only $y$ occurs in $v_{1},\ldots,v_{n}$ and $R^{v_{1},\ldots,v_{n}}$ occurs in $\varphi$ or $\psi$. 
\end{itemize} 
$\Bmf'$ is defined in the same way using $\Bmf$. It is readily checked that
$\Amf',\abf \sim_{\text{FO}^{2},\tau'} \Bmf',\bbf$, $\Amf'\models \varphi^{\dagger}(\abf)$, and $\Bmf'\models \psi^{\dagger}(\bbf)$.

Conversely, assume $\Amf,\abf \sim_{\text{FO}^{2},\tau'} \Bmf,\bbf$, $\Amf\models \varphi^{\dagger}(\abf)$, and $\Bmf\models \psi^{\dagger}(\bbf)$. We define the structure $\Amf'$ in the same way as $\Amf$ except that for relation symbols $R$ of arity $\geq 3$:
\begin{itemize}
	\item $(v(v_{1}),\ldots,v(v_{n}))\in R^{\Amf'}$ if $R(v_{1},\ldots,v_{n})$ occurs in $\varphi$ or $\psi$ such that for the assignment $v$
	it holds that $\Amf\models_{v} R^{v_{1},\ldots,v_{n}}(x,y)$ (or, if only $x$ or only $y$ occur in $v_{1},\ldots,v_{n}$, $\Amf\models_{v} R^{v_{1},\ldots,v_{n}}(x)$ or $\Amf\models_{v} R^{v_{1},\ldots,v_{n}}(y)$ respectively).
	\item no other tuples are in $R^{\Amf'}$.
\end{itemize} 
$\Bmf'$ is defined in the same way using $\Bmf$.
Using the conjuncts $\chi_{R}$ it can be shown that
$\Amf',\abf \sim_{\text{FO}^{2},\tau} \Bmf',\bbf$ and $\Amf'\models \varphi(\abf)$ and $\Bmf'\models \psi(\bbf)$.			
\end{proof}

We first show that one can achieve Condition~(P2) for links between
semi $i$-pawns. Recall  that for any king mosaics 
$m_{1}, m_{2}$ and $i$-pawns 
$t_{1}\in m_{1}$ and $t_{2}\in m_{2}$ 
\begin{eqnarray*}
	L_{i}((t_{1},m_{1}),(t_{2},m_{2})) & = & 
	\{ l_{\Amf_{i}}(d_{1},d_{2}) \mid d_{1}\not=d_{2}, \\
	& &  \;\;\; t_{1}= t_{\Amf_{i}}(d_{1}), m_{1}=m(d_{1}),\\
	& &  \;\;\; t_{2}= t_{\Amf_{i}}(d_{2}), m_{2}=m(d_{2})\}.	
\end{eqnarray*}
We aim to define link-types $l_{\Bmf_{i}}((t_{1}',m_{1}),(t_{2}',m_{2}))$, where $t_{1}'$ and $t_{2}'$ are copies of $t_{1}$ and $t_{2}$, in such a way that the following holds:
\begin{itemize}
	\item[(P2)] If $t_{1}'$ is a copy of $t_{1}$ and $l$ a link-type, then there exists a copy $t_{2}'$ of $t_{2}$ such that $l=l_{\Bmf_{i}}((t_{1}',m_{1}),(t_{2}',m_{2}))$ iff $l\in L_{i}((t_{1},m_{1}),(t_{2},m_{2}))$.   
\end{itemize}
In the construction, we use the fact that there are $\geq 4\times 2^{2s}$ many copies of any $i$-pawn and that the number of link-types does not exceed $2^{2s}$.
Assume first that $(t_{1},m_{1})\not=(t_{2},m_{2})$. Then partition, for $j=1,2$, the set $\{(t_{j},1),\ldots,(t_{j},k_{1})\}$ of copies of $t_{j}$ into two sets  $M_{1}^{j},M_{2}^{j}$ such that $|M_{1}^{j}|,|M_{2}^{j}| \geq 2^{2s}$.
Now we define the link-types between any pair $(t_{1}',m_{1})$ and $(t_{2}',m_{2})$ as follows

\begin{itemize}
	\item for every $t_{1}'\in M_{1}^{1}$ do the following: take for any 
	link-type $l\in L_{i}((t_{1},m_{1}),(t_{2},m_{2}))$ some $(t_{2}',m_{2})$ with $t_{2}'\in M_{1}^{2}$ and set
	$$
	l_{\Bmf_{i}}((t_{1}',m_{1}),(t_{2}',m_{2})):= l
	$$
	There are sufficiently many $(t_{2}',m_{2})$ with $t_{2}'\in M_{1}^{2}$
	since $|M_{1}^{2}|\geq 2^{2s}$.	
	\item for every $t_{1}'\in M_{2}^{1}$ do the following: take for any 
	link-type $l\in L_{i}((t_{1},m_{1}),(t_{2},m_{2}))$ some $(t_{2}',m_{2})$ with $t_{2}'\in M_{2}^{2}$ and set
	$$
	l_{\Bmf_{i}}((t_{1}',m_{1}),(t_{2}',m_{2})):= l
	$$
	There are sufficiently many $(t_{2}',m_{2})$ with $t_{2}'\in M_{2}^{2}$
	since $|M_{2}^{2}|\geq 2^{2s}$.	
	\item for every $t_{2}'\in M_{1}^{2}$ do the following: take for any 
	link-type $l\in L_{i}((t_{1},m_{1}),(t_{2},m_{2}))$ some $(t_{1}',m_{1})$ with $t_{1}'\in M_{2}^{1}$ and set
	$$
	l_{\Bmf_{i}}((t_{1}',m_{1}),(t_{2}',m_{2})):= l
	$$
	There are sufficiently many $(t_{1}',m_{1})$ with $t_{1}'\in M_{2}^{1}$
	since $|M_{2}^{1}|\geq 2^{2s}$.	
	\item for every $t_{2}'\in M_{2}^{2}$ do the following: take for any
	 link-type $l\in L_{i}((t_{1},m_{1}),(t_{2},m_{2}))$ some $(t_{1}',m_{1})$ with $t_{1}'\in M_{1}^{1}$ and set
	$$
	l_{\Bmf_{i}}((t_{1}',m_{1}),(t_{2}',m_{2})):= l
	$$
	There are sufficiently many $(t_{1}',m_{1})$ with $t_{1}'\in M_{1}^{1}$
	since $|M_{1}^{1}|\geq 2^{2s}$.	
\end{itemize}   
For semi $i$-pawns $(t_{1}',m_{1}),(t_{2}',m_{2})$ that have not yet been connected by any of the four steps above, choose an arbitrary link-type $l$ from $L_{i}((t_{1},m_{1}),(t_{2},m_{2}))$ and set
$l_{\Bmf_{i}}((t_{1}',m_{1}),(t_{2}',m_{2})):= l$. It is readily checked that
(P2) is satisfied. 

Now assume that $(t_{1},m_{1})=(t_{2},m_{2})$. Then partition the set $\{(t_{1},1),\ldots,(t_{1},k_{1})\}$ of copies of $t_{1}$ into four sets  $M_{1}^{j},M_{2}^{j}$ such that $|M_{1}^{j}|,|M_{2}^{j}| \geq 2^{2s}$, $j=1,2$,
and define $l_{\Bmf_{i}}((t_{1},k),m_{1}),((t_{1},k'),m_{1}))$ in exactly
the same way as above for $(t_{1},k),(t_{1},k')\in M_{r}^{1}\times M_{r'}^{2}$,
$r,r'\in \{1,2\}$. For any $((t_{1},k),m_{1}),((t_{1},k'),m_{1}))$ with $k\not=k'$ for which $l_{\Bmf_{i}}((t_{1},k),m_{1}),((t_{1},k'),m_{1}))$
has not yet been defined choose an arbitrary link-type $l$ from $L_{i}((t_{1},m_{1}),(t_{2},m_{2}))$ and set
$l_{\Bmf_{i}}(((t_{1},k),m_{1}),((t_{1},k'),m_{2})):= l$.
Then (P2) is satisfied.

\medskip

We now show that Conditions~(1) and (2) are satisfied, starting with Condition~(1). 

\begin{lemma}\label{lem:var2lem1}
	Let $t'$ be a copy of $t$ and $m'$ a copy of $m$. For $i=1,2$, all $(t',m')\in \text{dom}(\Bmf_{i})$, the witness $d$ of $(t',m')$ in $\Amf_{i}$, and all $\gamma(x)\in \text{cl}(\Xi)$: 
	$$
	\Bmf_{i} \models \gamma(t',m') \quad \Leftrightarrow 
	\quad \gamma(x) \in t_{\Amf_{i}}(d)\quad \Leftrightarrow \quad \gamma(x) \in t.
	$$
\end{lemma}
\begin{proof}\
	The equivalence `$\gamma(x) \in t_{\Amf_{i}}(d)$ iff $\gamma(x) \in t$' follows from the definition of witnesses $d$ of $(t',m')$. We thus show the first equivalence.  
	For $\gamma(x)$ of the form $R(x)$ or $R(x,x)$ the equivalence holds by definition. It thus suffices to show the first equivalence for  existentially quantified $\gamma(x)=\exists y \beta(x,y)$ with $\beta(x,y)$ quantifier-free.
	
	\medskip
	
	($\Rightarrow$) It suffices to observe that the following holds for all $(t_{1}',m_{1}')\in \text{dom}(\Bmf_{i})$: if $l=l_{\Bmf_{i}}((t_{1}',m_{1}'),(t_{2}',m_{2}'))$ for some $(t_{2}',m_{2}')\in \text{dom}(\Bmf_{i})$, then there exist $d_{1},d_{2}$ with $m_{j}=m(d_{j})$
	and $t_{j}=t_{\Amf_{j}}(d_{j})$ for $j=1,2$ such that $l=l_{\Amf_{i}}(d_{1},d_{2})$.
	
	\medskip

	($\Leftarrow$) We show the following

\medskip
\noindent	
	Claim 1. Let $d_{1}$ be the witness for $(t_{1}',m_{1}')$. If  $l=l_{\Amf_{i}}(d_{1},d_{2})$ for some $d_{2}\in \text{dom}(\Amf_{i})$, then there exists $(t_{2}',m_{2}')$ such that $m_{2}=m(d_{2})$, $t_{2}= t_{\Amf_{i}}(d_{2})$, and $l=l_{\Bmf_{i}}((t_{1}',m_{1}'),(t_{2}',m_{2}'))$. 
	
	\medskip
	For the proof of Claim~1 let $d_{1}$ be the witness for $(t_{1}',m_{1}')$ and  $l=l_{\Amf_{i}}(d_{1},d_{2})$ for some $d_{2}\in \text{dom}(\Amf_{i})$.
	
	\medskip
	
	Case 1. $t_{\Amf_{i}}(d_{1})$ is an $i$-king in king mosaic $m=m(d_{1})$. 
	
	If $t_{\Amf_{i}}(d_{2})$ is an $i$-king in king mosaic $m(d_{2})$, then  $(t_{\Amf_{i}}(d_{2}),m(d_{2}))$ is as required. 
	
	If $t_{\Amf_{i}}(d_{2})$ is an $i$-pawn in king mosaic $m(d_{2})$, then by (P1) there exists $(t_{2}',m(d_{2}))$ such that $t_{2}= t_{\Amf_{i}}(d_{2})$ and $l=l_{\Bmf_{i}}((t_{\Amf_{i}}(d_{1}),m(d_{1})),(t_{2}',m(d_{2}))$.
	Then $(t_{2}',m(d_{2}))$ is as required.
	
	If $t_{\Amf_{i}}(d_{2})$ is an $i$-pawn in pawn mosaic $m(d_{2})$, then by (M) and (P3) there exists a full $i$-pawn $(t_{2}',m')$ such that $t_{2}= t_{\Amf_{i}}(d_{2})$ and $l=l_{\Bmf_{i}}((t_{\Amf_{i}}(d_{1}),m(d_{1})),(t_{2}',m'))$, as required.
	
	\medskip
	
	Case 2. $(t_{1}',m_{1}')$ is a semi $i$-pawn. The claim follows from (P1) if $t_{\Amf_{i}}(d_{2})$ is an $i$-king in king mosaic $m(d_{2})$. For $t_{\Amf_{i}}(d_{2})$ an $i$-pawn in king mosaic $m(d_{2})$, the claim follows from (P2). For $t_{\Amf_{i}}(d_{2})$ an $i$-pawn in a pawn mosaic $m(d_{2})$, the claim follows from (P4).

	\medskip
	
	Case 3. $(t_{1}',m_{1}')$ is a full $i$-pawn.
	The claim follows from (P3) if $t_{\Amf_{i}}(d_{2})$ is an $i$-king in king mosaic $m(d_{2})$. For $t_{\Amf_{i}}(d_{2})$ an $i$-pawn in king mosaic $m(d_{2})$, the claim follows from (P4). For $t_{\Amf_{i}}(d_{2})$ an $i$-pawn in a pawn mosaic $m(d_{2})$, the claim follows from (P5).
\end{proof}
We now prove Condition~(2). The \emph{restriction} $l_{|\tau}$ of a link-type $l$ to a signature $\tau$ is the set of all $R(x,y)$ and $R(y,x)$ in $l$ with $R\in \tau$. Any such restriction is called a \emph{$\tau$-link}.
%
%
%

Note first that for all $(t_{1}',m'), (t_{2}',m')\in \text{dom}(\Bmf_{1})\cup \text{dom}(\Bmf_{2})$ there exists a generator $g$ of $m'$ and witnesses $d_{1}$ for $(t_{1}',m')$ and $d_{2}$ for $(t_{2}',m')$. Thus $d_{1}\sim_{\text{FO}^{2},\tau} d_{2}$ and $t_{i}=d_{\Amf_{j}}(d_{i})$ for appropriate $j\in \{1,2\}$. Thus 	
$\chi(x)\in t_{1}$ iff $\chi(x)\in t_{2}$ for any formula $\chi(x)$ of the form $R(x)$ or $R(x,x)$ with $R\in \tau$. We obtain from Lemma~\ref{lem:var2lem1}
that $(t_{1}',m')$ and $(t_{2}',m')$ satisfy the same atomic formulas $R(x)$ and $R(x,x)$ with $R\in \tau$. To fully check Conditions~(i) and (ii) for FO$^{2}(\tau)$-bisimulations, we introduce some notation.
For $(t_{1}',m_{1}')\in \text{dom}(\Bmf_{i})$, any copy $m_{2}'$ of a mosaic, and any $\tau$-link $h$, we write 
$$
(t_{1}',m_{1}') \rightarrow_{h} m_{2}'
$$
if there exists $(t_{2}',m_{2}')\in \text{dom}(\Bmf_{i})$ such that 
$$
h=l_{\Bmf_{i}}((t_{1}',m_{1}'),(t_{2}',m_{2}'))_{|\tau} 
$$
Then it suffices to show the following
\begin{lemma}
	For $(t_{1}',m_{1}'), (t_{2}',m_{1}')\in \text{dom}(\Bmf_{1})\cup \text{dom}(\Bmf_{2})$, any copy $m_{2}'$ of a mosaic, and any $\tau$-link $h$: 
	$$
	(t_{1}',m_{1}') \rightarrow_{h} m_{2}' \quad \Leftrightarrow \quad
	(t_{2}',m_{1}') \rightarrow_{h} m_{2}'
	$$
\end{lemma} 
%

\begin{proof}
	We make a case distinction and first assume that $m_{1}$ or $m_{2}$ is a king mosaic. By construction, then $(t_{1}',m_{1}') \rightarrow_{h} m_{2}'$ implies that for the generators $g_{1},g_{2}$ of $m_{1}',m_{2}'$
	there are $d_{1} \sim_{\text{FO}^{2},\tau} g_{1}$ and $d_{2} \sim_{\text{FO}^{2},\tau} g_{2}$ such that 
	$h= l_{\Amf_{i}}(d_{1},d_{2})_{|\tau}$, for appropriate $i$.
	But then, by the definition of mosaics and FO$^{2}(\tau)$-bisimulations, for all 
	$d_{1} \sim_{\text{FO}^{2},\tau} g_{1}$ there exists $d_{2} \sim_{\text{FO}^{2},\tau} g_{2}$ such that 
	$h= l_{\Amf_{i}}(d_{1},d_{2})_{|\tau}$, for appropriate $i$.
	But then $(t_{2}',m_{1}') \rightarrow_{h} m_{2}'$ for all 
	$(t_{2}',m_{1}')$.
	
	\medskip
	
	Now assume that $m_{1}$ and $m_{2}$ are both pawn mosaics. By construction,
	then $(t_{1}',m_{1}') \rightarrow_{h} m_{2}'$ iff 
	there are $d_{1},d_{2}$ with $m_{1}=m(d_{1})$ and $m_{2}=m(d_{2})$ such that 
	$t_{1}=t_{\Amf_{i}}(d_{1})$ and $h= l_{\Amf_{i}}(d_{1},d_{2})_{|\tau}$, for appropriate $i$.
	But by the definition of mosaics and FO$^{2}(\tau)$-bisimulations, the latter is the case 
	iff there are $d_{1},d_{2}$ with $m_{1}=m(d_{1})$ and $m_{2}=m(d_{2})$ such that $h= l_{\Amf_{i}}(d_{1},d_{2})_{|\tau}$, for appropriate $i$.
	This condition does not depend on $t_{1}'$ and so $(t_{2}',m_{1}') \rightarrow_{h} m_{2}'$ follows.
\end{proof}

\lemcorrectfotwo*

\begin{proof}
  The proof is essentially the same as the proof for
  Lemma~\ref{lem:correct2exp}; we give it here for the sake of
  completeness. 

  \smallskip
  $(\Rightarrow)$ If $M$ accepts $w$, there is a computation tree of $M$ on $w$. We construct a
single model $\Amf$ of $\varphi'$ as follows. Let $\Amf^*$ be the infinite tree-shaped structure
that represents the computation tree of $M$ on $w$ as described above,
that is, configurations are represented by sequences of $2^n$ elements
linked by $S$. Moreover, all elements of a configuration are labeled
with $B_\forall$, $B_\exists^1$, or $B_\exists^2$ depending on whether
the configuration is universal or existential, and in the latter case
the superscript indicates which choice has been made for the
existential state. Finally, the first element of the first successor
configuration of a universal configuration is labeled with $Z$.
In particular, 
$\Amf^*$ only interprets the symbols in $\tau$ non-empty. Now, we
obtain structures $\Amf_k$, $k<2^n$ from $\Amf^*$ by interpreting
non-$\tau$-symbols as follows: 
\begin{itemize}
	
  \item the entire domain of $\Amf_{k}$ satisfies $I$;


  \item the $U$-counter starts at $0$ at the root and counts modulo
    $2^n$ along each $S$-path;

  \item the $V$-counter starts at $k$ at the root and
    counts modulo $2^n$ along each $S$-path;

  \item the auxiliary concept names of the shape $A_\sigma^i$ and
    $A_\sigma'$ are interpreted in a minimal way so as to satisfy the
    sentences listed above. Note that the sentences are Horn, thus
    there is no choice.

\end{itemize}
Now obtain $\Amf$ from $\Amf^*$ and the $\Amf_k$ as follows.  First,
create a both side infinite $R$-path \[\ldots
b_{-2}Rb_{-1}Rb_0Rb_1Rb_2\ldots \] and realize the corresponding
$A$-counter along the path and label every $b_{3k}$, $k\in\mathbb{Z}$,
with $X$. Then, add all $\Amf_k^*$ to every node $b_{3k}$, $k\in
\mathbb{Z}$, on the path by identifying the roots of the $\Amf_k$ with
the respective node on the path.  Moreover, add to $\Amf$ three
elements $a_0,a_1,a_2$ such that 
$(a_0,a_1),(a_1,a_2),(a_2,a_0)\in R^\Amf$ $a_0\in X^\Amf$, $a_0\in
Y^\Amf$, and $a_0\in
A^\Amf$. Finally, add a copy of $\Amf^*$ to \Amf by identifying the
root of $\Amf^*$ with $a_0$. We claim that $\Amf$ is as required. In
particular, $\Amf,a_0$ is a model of $\varphi\wedge A(x)$, $\Amf,b_0$
is a model of $\varphi\wedge \neg A(x)$, and the set $S$ of all
pairs
\begin{itemize}

  \item $(a_i,b_{i+3k})$ with
    $k\in\mathbb{Z}$, $i\in\{0,1,2\}$, and 

  \item $(e,e')$ with $e'$ copy of $e$ in some $\Amf_k$,

\end{itemize}
is an FO$^2(\tau)$-bisimulation on $\Amf$ with $(a_0,b_0)\in S$.

$(\Leftarrow)$ Let $\Amf,\Bmf$ be a models of $\varphi$ such that
$\Amf,a\sim_{\text{FO}^2,\tau}\Bmf,b$ for some elements $a,b$ with
$a\in A^\Amf$, $b\notin A^{\Bmf}$. As it was argued above, due to the
three-element $R$-loop enforced at $a$ via $\varphi_0'$,
from $b$ there has to be an
outgoing infinite $R$-path on which every third element
FO$^2(\tau)$-bisimilar, and thus the $S$-trees starting at these
elements are also FO$^2(\tau)$-bisimilar.
(There is also an incoming infinite $R$-path with
this property, but it is not relevant for the proof.) All those
$S$-trees are additionally labeled with some auxiliary relation
symbols not in $\tau$, depending on the distance from $b$. However, it
can be shown using the arguments that accompanied the construction of
$\varphi'$ that all $S$-trees contain a computation tree of $M$ on
input $w$. Hence, $M$ accepts $w$.
\end{proof}
Recall the definition of $\varphi''$ and $\tau'$ from the paper. Then it suffices to prove the following.

\begin{restatable}{lemma}{lemcorrectfotwosig}\label{lem:correct2exp-fo2sig}
	$M$ accepts the input $w$ iff there exists models $\Amf,\Bmf$ of
	$\varphi''$
	and elements $a\in A^\Amf$, $b\notin A^\Bmf$ 
	such that $\Amf,a\sim_{\text{FO}^2,\tau'}\Bmf,b$.
\end{restatable}
\begin{proof}
	($\Leftarrow$) is an immediate consequence of Lemma~\ref{lem:correct2exp-fo2}.
	
	\medskip
	
	($\Rightarrow$) We expand the model $\Amf$ constructed in the proof of Lemma~\ref{lem:correct2exp-fo2} and obtain a model $\Amf'$ of $\varphi''$
	such that $a_0\in A^{\Amf'}$, $b_0\notin A^{\Amf'}$ 
	and $\Amf',a_0\sim_{\text{FO}^2,\tau'}\Amf',b_0$.
	
In detail, to define $\Amf'$ we keep $\Amf$ but attach to every $X$-element $d$ of $\Amf$ an $N^{-1}$ path of length $2$ from $d$ to a fresh node $(d,E)$,
for every $E\in \text{sig}(\varphi')\setminus (\tau \cup \{A\})$. 
Now we proceed as follows for every $E\in
\text{sig}(\varphi')\setminus (\tau \cup \{A\})$: add
$(d,(a_{0},E))$ to $R_{E}^{\Amf'}$ iff $d\in E^{\Amf}$, for all $d\in \text{dom}(\Amf)$. 
This ensures that $\Amf'\models \chi_{E}(d)$ iff $d\in E^{\Amf}$ for all such $d$.
Let $\Delta_{0},\Delta_{1},\ldots$ be the maximal subsets of $\text{dom}(\Amf)$ such that all elements of $\Delta_{i}$ are 
FO$^{2}(\tau)$-bisimilar in $\Amf$. We add additional pairs to $R_{E}^{\Amf'}$ in such a way 
that all elements of any $\Delta_{i}$ are also FO$^{2}(\tau')$-bisimilar in $\Amf'$:
\begin{itemize}
	\item if $\Delta_{i}\supseteq E^{\Amf}$ then also add
	$(d,(d',E))$ to $R_{E}^{\Amf'}$ for all $X$-elements $d'$ and $d\in \Delta_{i}$;
	\item if $\Delta_{i}\cap  E^{\Amf}=\emptyset$, then we do not add any $(d,(d',E))$ to $R_{E}^{\Amf'}$, for $X$-elements $d'$ and $d\in \Delta_{i}$;
	\item otherwise we make sure that for every $X$-element $d'$ there exist both
	$d\in \Delta_{i}$ with $(d,(d',E))\in R_{E}^{\Amf'}$ and $e\in
	\Delta_{i}$ with $(e,(d',E))\not\in R_{E}^{\Amf}$ and we make
	sure for every $d\in \Delta_{i}$ there exist both an
	$X$-element $d'$ with $(d,(d',E))\in R_{E}^{\Amf'}$ and an
	$X$-element $e'$ with $(d,(e',E))\not\in R_{E}^{\Amf'}$. This
	is easily achieved without adding any additional pairs of the
	form $(d,(a_{0},E))$ to $R_{E}^{\Amf'}$.
\end{itemize}    
This finishes the definition of $\Amf'$. It is easy to see that $\Amf'$ is as required.   
\end{proof}

\end{document}